%% file: Grassmann0309.tex
\documentclass[12pt,journal,draftcls,letterpaper,onecolumn]{IEEEtran}

\input deflatex.tex

\usepackage{url}
\usepackage{graphicx,subfig}
\usepackage{epsfig}
\usepackage{amsmath}
\usepackage{amssymb}
\usepackage{psfrag}

\def\w{{\bf w}}

\def\v{{\bf v}}
\def\x{{\bf x}}
\def\y{{\bf y}}

\def\b{{\bf b}}

\def\x{{\mathbf x}}

\def\w{{\bf w}}

\def\v{{\bf v}}
\def\x{{\bf x}}
\def\y{{\bf y}}

\def\b{{\bf b}}

\def\real{{\rm Re}\,}

\def\be{\begin{equation}}
\def\ee{\end{equation}}
\def\ba{\left[\begin{array}}
\def\ea{\end{array}\right]}

\def\w{{\bf w}}

\def\v{{\bf v}}
\def\x{{\bf x}}
\def\y{{\bf y}}

\def\b{{\bf b}}

\begin{document}

\title{Compressive Sensing over the Grassmann Manifold: a Unified Geometric Framework}


\author{Weiyu~Xu
        and Babak~Hassibi
\thanks{Weiyu Xu was with the Department
of Electrical Engineering, California Institute of Technology, Pasadena,
CA, 91125 USA. e-mail: weiyu@systems.caltech.edu. He is now with School of Electrical and Computer Engineering, Cornell University.}
\thanks{Babak Hassibi is with the Department
of Electrical Engineering, California Institute of Technology, Pasadena,
CA, 91125 USA. e-mail: hassibi@systems.caltech.edu.}
}

\maketitle

\begin{abstract}

$\ell_1$ minimization is often used for finding the sparse solutions of an under-determined linear system. In this paper we focus on finding sharp performance bounds on recovering approximately sparse signals using $\ell_1$ minimization, possibly under noisy measurements. While the restricted isometry property is powerful for the analysis of recovering approximately sparse signals with noisy measurements, the known bounds on the achievable sparsity\footnote{The ``sparsity" in this paper means the size of the set of nonzero or significant elements in a signal vector.} level can be quite loose. The neighborly polytope analysis which yields sharp bounds for ideally sparse signals cannot be readily generalized to approximately sparse signals. Starting from a necessary and sufficient condition, the ``balancedness" property of linear subspaces, for achieving a certain signal recovery accuracy, we give a unified \emph{null space Grassmann angle}-based geometric framework for analyzing the performance of $\ell_1$ minimization. By investigating the ``balancedness" property, this unified framework characterizes sharp quantitative tradeoffs between the considered sparsity and the recovery accuracy of the $\ell_{1}$ optimization. As a consequence, this generalizes the neighborly polytope result for ideally sparse signals. Besides the robustness in the ``strong" sense for \emph{all} sparse signals, we also discuss the notions of ``weak" and ``sectional'' robustness. Our results concern fundamental properties of linear subspaces and so may be of independent mathematical interest.

\end{abstract}

\section{Introduction}
Compressive sensing is an area in signal processing which has attracted a lot of attention recently
\cite{CandesCS} \cite{DonohoCS}. The motivation behind compressive
sensing is to do ``sampling" and ``compression" at the same time. In
conventional wisdom, in order to fully recover a signal, one has to
sample the signal at a sampling rate equal or greater to the Nyquist
sampling rate. The process of ``sampling at full rate" and then ``throwing away in
compression" can prove to be wasteful of sensing and sampling
resources, especially in application scenarios where resources like
sensors, energy, observation time, etc. are limited. However, in many applications such as imaging, sensor
networks, astronomy, biological systems \cite{RICE}, the signals of interest are often ``sparse" over a certain basis. In these
cases, compressive sensing promises to use a much smaller number of samples or measurements while still being able to recover the
original sparse signal exactly or approximately. What enables practical compressive sensing is the existence of efficient
decoding algorithms to recover the sparse signals from the
``compressed" measurements. One of the most import and
powerful decoding algorithms is the Basis Pursuit method, namely
the $\ell_{1}$ minimization method \cite{Claerbout73,Donoho06}.

In this paper, we are interested in analyzing the decoding performance of the
$\ell_{1}$ minimization algorithm for approximately sparse signals under possibly noisy measurements.
Mathematically, in compressive sensing problems, we would like to
find an $n \times 1$ vector $\x$ such that
\begin{equation}
\y=A\x, \label{eq:Grasssystem}
\end{equation}
where $A$ is an $m\times n$ measurement matrix, $\y$ is an $m\times
1$ measurement vector and $m<n$ in general. In the usual compressive sensing context $\x$
is an $n\times 1$ unknown $k$-sparse vector, which has
only $k$ nonzero components. In this paper we will consider a more
general version of the $k$-sparse vector $\x$. Namely, we will
assume that $k$ components of the vector $\x$ have large magnitudes
and that the vector comprised of the remaining $(n-k)$ components has
an $\ell_{1}$-norm less than some value, say, $\Delta$. We will refer to this type of
signal as an approximately $k$-sparse signal, or for brevity only
an approximately sparse signal. It is also possible that the $\y$ can be further
corrupted with measurement noise. This problem setup is more
realistic of practical applications than the standard compressive
sensing of ideally $k$-sparse signals (see, e.g., \cite{Tropp,CandesCS,CRT06} and the references therein). The interested readers can find more on similar type of problems in \cite{devore2} and other references.

In the rest of the paper we will further assume that the number of
the measurements is $m=\delta n$ and the number of the ``large''
components of $\x$ is $k=\rho\delta n=\zeta n$, where $0<\rho<1$ and
$0<\delta<1$ are constants independent of $n$ (clearly, $\delta
>\zeta$).

\subsection{$\ell_1$ Minimization for Ideally Sparse Signal}

$\ell_{1}$ minimization optimization (Basis Pursuit) proposes
solving the following problem
\begin{eqnarray}
\mbox{min} & & \|\x\|_{1}\nonumber \\
\mbox{subject to} & & \y=A\x, \label{eq:Grassl1}
\end{eqnarray}
where $\|\x\|_1$ denotes the $\ell_1$ norm of $\x$, namely the sum of the amplitudes of all the elements in $\x$.

In \cite{Candes05} the authors were able to show
that if the number of the measurements is $m=\delta n$ and if the
matrix $A$ satisfies a special property called the restricted
isometry property (RIP), then any unknown vector $\x$ with no more
than $k=\zeta n$ (where $\zeta$ is an absolute constant as a
function of $\delta$, but independent of $n$, and explicitly bounded
in \cite{Candes05}) nonzero elements can be recovered by solving
(\ref{eq:Grassl1}). As expected, this assumes that $\y$ was in fact
generated by such an $\x$ and given to us (more on the case when the
available measurements are noisy versions of $\y$ can be found in
e.g. \cite{hano06,wainwright}).

As can be immediately seen, the previous results heavily rely on the
assumption that the measurement matrix $A$ satisfies the RIP
condition. It turns out that for several specific classes of
matrices, such as matrices with independent zero-mean Gaussian
entries or independent Bernoulli entries, the RIP holds with
overwhelming probability \cite{Candes05,Wakin07,Ver}. However, it
should be noted that the RIP condition is only a sufficient condition for
$\ell_{1}$-optimization to produce a solution of
(\ref{eq:Grasssystem}).

Instead of characterizing the $m\times n$ matrix $A$ through the RIP
condition, in \cite{Donoho06,DT1}, the authors proposed to study $A$ through a $k$-neighborly polytope condition. As shown in \cite{Donoho06}, this characterization of the matrix $A$ is in fact a necessary and
sufficient condition for (\ref{eq:Grassl1}) to produce the sparse solution $\x$ satisfying
(\ref{eq:Grasssystem}). Furthermore, developing the results of \cite{Vershik92}, it can be shown that if the matrix $A$ has i.i.d.
zero-mean Gaussian entries, then the $k$-neighborly polytope condition holds with overwhelming probability. The precise relation between $m$, $n$ and $k$ in order for this to happen is characterized in
\cite{Donoho06}. It should also be noted that for a given
value $m$,  i.e. for a given value of the constant $\delta$, the
value of the constant $\zeta$ given by the neighborly polytope condition is significantly better in
\cite{Donoho06,DT1} than in \cite{Candes05}. In fact, the values
of $\zeta$ for the so-called ``weak" threshold, obtained for different values of $\delta$ in
\cite{Donoho06}, approach the ones obtained by simulation as
$n\rightarrow\infty$.

\subsection{$\ell_1$ Minimization for Approximately Sparse Signal}
As mentioned earlier, in this paper we will be interested in
recovering not perfectly $k$-sparse signals from compressed
observations $\y$. In this case an exact recovery of the unknown vector
$\x$ from a reduced number of measurements is not possible in
general. Instead, we will prove that, if we denote the unknown
signal as $\x$, denote $\hat{\x}$ as one
solution to (\ref{eq:Grassl1}), then for any given constant
$0<\delta< 1$ and any given constant $C>1$ (representing how close in $\ell_{1}$ norm the recovered vector
$\hat{\x}$ should be to $\x$), there exists a constant $\zeta>0$ and a sequence
of measurement matrices $A \in \mathbb{R}^{m \times n}$ as $n \rightarrow
\infty$ such that
\begin{equation}
||\hat{\x}-\x||_1\leq \frac{2(C+1)\Delta}{C-1}, \label{eq:noiseadj}
\end{equation}
holds for \emph{all }$\x \in \mathbb{R}^{n}$, where $\Delta$ is the $\ell_1$ norm of any $(n-k)$ elements of the vector $\x$ (recall $k=\zeta n$). Here $\zeta$ will be a function of
$C$ and $\delta$, but independent of the problem dimension $n$. In
particular, we have the following theorem.

\begin{theorem}
Let $n$, $m$, $k$, $\x$, $\hat{\x}$ and $\Delta$ be defined as
above. Let $K$ denote a subset of $\{1,2,\dots,n\}$ such that
$|K|=k$, where $|K|$ is the cardinality of $K$, and let $K_i$ denote
the $i$-th element of $K$ and $\overline{K}=\{1,2,\dots,n\} \setminus K$.

Then for any constant $C>1$ and any $\delta=\frac{m}{n}>0$, there
exists a $\zeta(\delta, C)>0$ such that if the measurement matrix
$A$ is the basis for a uniformly-distributed subspace, then with
overwhelming probability as $n \rightarrow \infty$, for all vectors
$\w \in \mathbb{R}^{n}$ in the null space of $A$, and for all $K$ such that
$|K|=k \leq \zeta(\delta, C)n$, we have
\begin{equation}
 C \sum_{i=1}^{k}|\w_{K_i}|\leq\sum_{i=1}^{n-k}|\w_{\overline{K}_i}|,
\label{eq:startthmeq}
\end{equation} where $\x_{K}$ denotes the part of $\x$ over the subset
$K$; and at the same time the solution $\hat{\x}$ produced by
(\ref{eq:Grassl1}) will satisfy

\begin{equation}
||\hat{\x}-\x||_1\leq \frac{2(C+1)\Delta}{C-1}.
\label{eq:startnoiseadjthm}
\end{equation}
for \emph{all} $\x \in \mathbb{R}^{n}$.
\end{theorem}

The main focus and contribution of this paper is to establish a sharp relationship between $\delta$, $\zeta$ and $C$. For example, when $\delta=\frac{m}{n}$ varies, we have Figure \ref{fig:fullstability} showing the tradeoff between the signal sparsity $\zeta$ and the parameter $C$, which determines the robustness \footnote{The ``robustness" concept in this sense is often called the ``stability" in other papers, for example, \cite{CandesCS}.} of the $\ell_1$ minimization. The curve for $C=1$ matches the ``strong" threshold curve from \cite{Donoho06} for ideally sparse signal vectors .


\begin{figure*}[htb]
\centerline{\epsfig{figure=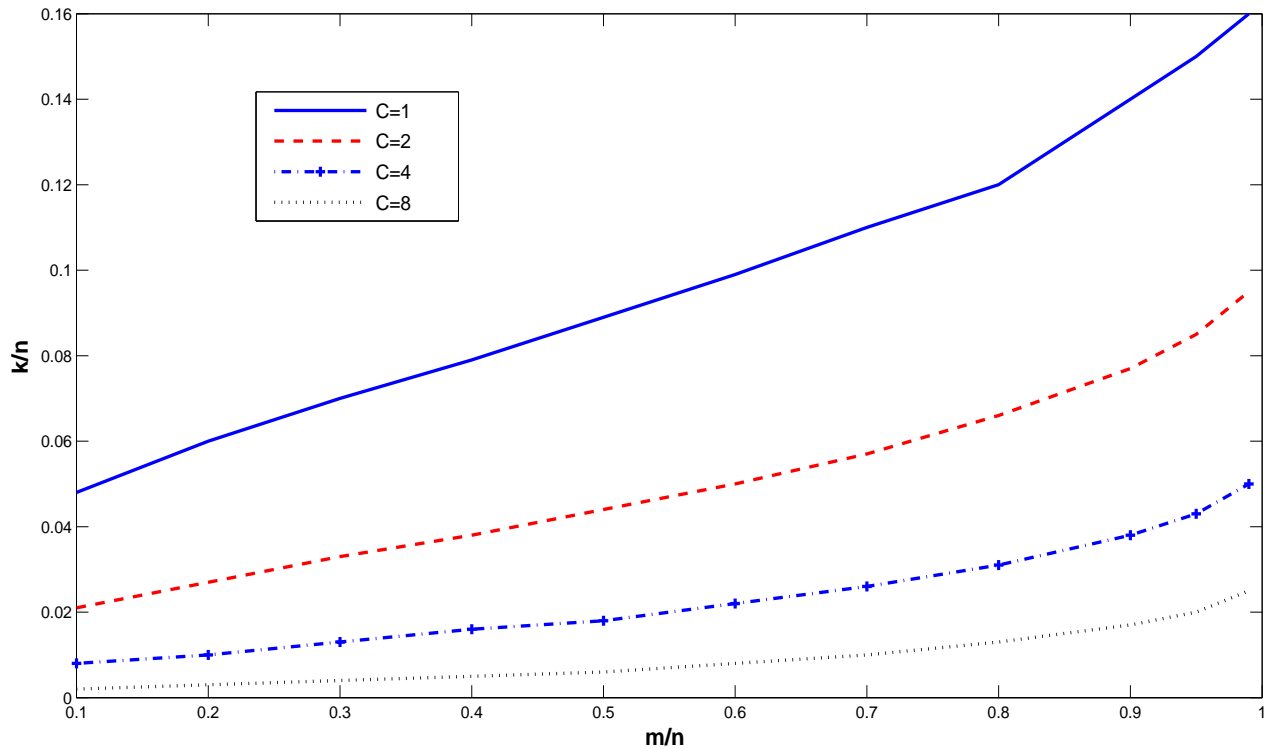,width=8 in,height=4.4 in}}
\caption{Tradeoff between signal sparsity and $\ell_1$ recovery robustness as a function of $C$ (allowable
imperfection of the recovered signal is $\frac{2(C+1)\Delta}{C-1}$)}
\label{fig:fullstability}
\end{figure*}

To obtain the stated results we will make use of a characterization
that constitutes both necessary and sufficient conditions on the
matrix $A$ such that the solution of (\ref{eq:Grassl1}) approximates
the original signal accurately enough such that (\ref{eq:noiseadj})
holds. This characterization will be equivalent to the neighborly
polytope characterization from \cite{Donoho06} in the ``ideally
sparse" case when $C=1$. Furthermore, as we will see later in the paper, in the
perfectly sparse signal case (which allows $C\rightarrow 1$),
our result for allowable $\zeta$ matches the result of
\cite{Donoho06}. Our analysis will be directly based on the
null space Grassmann angle approach in high dimensional integral
geometry, which gives a unified analytical framework for $\ell_{1}$
minimization.

A similar problem was considered in \cite{devore2}, where the null space characterization for recovering approximately sparse signal was analyzed using the RIP in \cite{Candes05}; however, no explicit values of $\zeta$ were given. Since the RIP condition is a sufficient condition for good sparse signal recoveries using $\ell_1$ minimization, it generally gives rather loose
bounds on the explicit values of $\zeta$ even in the ideally sparse signal case \cite{Candes05}\cite{BCT}. There have also been some recent works trying to analyze the performance of $\ell_1$ minimization through non-RIP techniques \cite{Y08,SV,Stojnic09}. Compared with previous results, in this paper we will provide sharp bounds on the explicit values of the allowable constants $\zeta$ for satisfying the subspace ``balancedness" condition, as a function of $C\geq 1$. In the literature, there are also discussions of compressive sensing under different definitions of non-ideally sparse signals, for example, \cite{DonohoCS} discusses compressive sensing for signals from a
$\ell_p$ ball with $0<p\leq 1$ using sufficient conditions based on
results of the Gelfand $n$-widths. However, the results of this
paper are dealing directly with approximately sparse signals defined
in terms of the concentration of $\ell_{1}$ norm, and furthermore,
we give a neat \emph{necessary and sufficient} condition for
$\ell_{1}$ optimization to be robust and we are also able to explicitly
give much sharper bounds on the sparsity parameter $\zeta$. When we finalize this draft from our earlier conference publication \cite{XuHassibiAllerton08}, we are informed of the very recent work \cite{do10} which deals with a related but different problem formulation of characterizing the tradeoff between signal sparsity and noise sensitivity of LASSO recovery method. Compared with \cite{do10}, we are dealing with the plain $\ell_1$ minimization method for recovering approximately sparse signals, and the performance bounds in this paper apply to general type of signals and noises. The analysis from \cite{do10} is an average-case analysis for compressed measurements corrupted with Gaussian noises, while the analysis in this paper provides both average-case and worst-case performance bounds under general types of signals and noises. It is also noteworthy pointing out that this work considers the plain $\ell_1$ minimization, which does not require the decoder to know of the statistical variance of the measurement noises. The analysis methodologies between this work and \cite{do10} are also different: this work relies on the analytical tools from the high dimensional polytope geometry, while \cite{do10} builds on the innovations of analyzing message passing algorithms.

The rest of the paper is organized as follows. In Section
\ref{sec:nullspace}, we introduce a null space characterization
of linear subspaces for guaranteeing robust signal recovery using the $\ell_{1}$ minimization. Section \ref{sec:Grassprobnull}
presents a Grassmann angle-based high dimensional geometrical
framework for analyzing the null space characterization. In Sections
\ref{sec:evathebound}, \ref{sec:bndexnangle}, and
\ref{sec:bndinternal}, analytical performance bounds are given for
the null space characterization. Section \ref{sec:wss} shows how the
Grassmann angle analytical framework can be extended to analyzing
the ``weak", ``sectional" and ``strong" notations of robust signal recovery.  In Section \ref{sec:noisymeas}, we present the
robustness analysis of the $\ell_{1}$ minimization under noisy
measurements. In Section
\ref{sec:numericalold}, the numerical evaluations of the performance
bounds for robust signal recovery are given. Section
\ref{sec:con} concludes the paper. In the appendix, we provide a quick summary of relevant geometric concepts in the high dimensional geometry and the proofs of related lemmas and theorems.

\section{The Null Space characterization}
\label{sec:nullspace}

In this section we introduce a useful characterization of the matrix
$A$. The characterization will establish a necessary and sufficient
condition on the matrix $A$ so that the solution of (\ref{eq:Grassl1})
approximates the solution of (\ref{eq:Grasssystem}) such that
(\ref{eq:noiseadj}) holds. (See \cite{FN,LN,Y,devore2,StXuHassibi08,StXuHaisit,Remark} etc. for
variations of this result).
\begin{theorem}
Assume that  $A$ is a general $m\times n$ measurement matrix. Let $C>1$ be a positive number. Further,
assume that $\y=A\x$  and that $\w$ is an $n\times 1$ vector. Let
$K$ be a subset of $\{1,2,\dots,n\}$ such that $|K|=k$, where
$|K|$ is the cardinality of $K$ and let $K_i$ denote the $i$-th
element of $K$. Further, let $\overline{K}=\{1,2,\dots,n\} \setminus K$.
Then for \emph{any} $\x \in \mathbb{R}^n$, for any $K$ such that $|K|=k$, \emph{any} solution $\hat{\x}$ produced by (\ref{eq:Grassl1}) will
satisfy
\begin{equation}
\|\x-\hat{\x}\|_{1} \leq \frac{2(C+1)}{C-1}  \| \x_{\overline{K}} \|_{1},
\label{eq:inthmbound}
\end{equation}
if $~\forall \w \in \mathbb{R}^n$ such that
\begin{equation*}
~A\w=0~
\end{equation*}
and $~\forall K$ such that $|K|=k$, we have
\begin{equation}
 C \sum_{i=1}^{k}|\w_{K_i}|\leq\sum_{i=1}^{n-k}|\w_{\overline{K}_i}|.
\label{eq:thmeq}
\end{equation}

Conversely, there exists some measurement matrix $A$, a set $K$ with cardinality $k$, an $\x$, and corresponding $\hat{\x}$ ($\hat{\x}$ is a minimizer to the programming (\ref{eq:Grassl1})),  such that (\ref{eq:thmeq}) is satisfied with equality for some vector $\w$ in the null space of $A$ with a constant $C'>1$; moreover
\begin{equation*}
\|\x-\hat{\x}\|_{1} =2\frac{(C'+1)}{C'-1}  \| \x_{\overline{K}} \|_{1}.
\end{equation*}

and

\begin{equation*}
\|\x-\hat{\x}\|_{1} >2\frac{(C+1)}{C-1}  \| \x_{\overline{K}} \|_{1}.
\end{equation*}

for any $C$ bigger than the constant $C'$.

\label{thm:th1}
\end{theorem}
\begin{proof} First, suppose the matrix $A$ has the claimed null space
property as in (\ref{eq:thmeq}) and we want to prove that any solution $\hat{\x}$ satisfies (\ref{eq:inthmbound}). Note that the solution $\hat{\x}$ of (\ref{eq:Grassl1})
satisfies
\begin{equation*}
\|\hat{\x}\|_{1} \leq \|\x\|_{1},
\end{equation*}
where $\x$ is the original signal. Since $A \hat{\x}=\y$, it easily
follows that $\w=\hat{\x}-\x$ is in the null space of $A$. Therefore
we can further write $\|\x\|_{1} \geq \| \x+\w \|_{1}$. Using  the
triangular inequality for the $\ell_{1}$ norm we obtain
\begin{eqnarray*}
\|\x_{K}\|_{1}+\|\x_{\overline{K}}\|_{1}&=& \|\x\|_{1}\nonumber\\
 &\geq& \| \hat{\x} \|_{1} =\| \x+\w \|_{1} \nonumber \\
 & \geq &\|\x_{K}\|_{1}-\|\w_{K}\|_{1}+\|\w_{\overline{K}}\|_{1}-\|\x_{\overline{K}}\|_{1}\nonumber \\
 & \geq
 &\|\x_{K}\|_{1}-\|\x_{\overline{K}}\|_{1}+\frac{C-1}{C+1}\|\w\|_{1}\nonumber
\end{eqnarray*}
where the last inequality is from the claimed null space property.
Relating the head and tail of the inequality chain above,
\begin{equation*}
2\|\x_{\overline{K}}\|_{1} \geq \frac{(C-1)}{C+1}\|\w\|_{1}.
\end{equation*}

Now we prove the second part of the theorem, namely when (\ref{eq:thmeq}) is violated, there exist scenarios where the error performance bound (\ref{eq:inthmbound}) fails. The simplest example is when the null space of the measurement matrix $A$ is a one-dimensional subspace and has an all-$1$ vector $(1,1, ..., 1)$ as its basis. Let $n$ be an even number. For any $k< \frac{n}{2}$, let us take $C'=\frac{n-k}{k}$ and $C=\frac{n-k}{k}+\epsilon$, where $\epsilon>0$ is an arbitrarily small positive number. Then obviously there exists a vector $\w$ in the null space of $A$ that violates the condition (\ref{eq:thmeq}) for $C=\frac{n-k}{k}+\epsilon$ for the set $K=\{1,2, ..., k\}$. Now we consider a signal vector
\begin{equation*}
 \x=(\underbrace{-1,-1,...,-1}_{\frac{n}{2}}, \underbrace{0,0,...,0}_{\frac{n}{2}}).
\end{equation*}
Taking the null space of $A$ into account, we can see
\begin{equation*}
 \hat{\x}=(\underbrace{0,0,...,0}_{\frac{n}{2}}, \underbrace{1,1,...,1}_{\frac{n}{2}})
\end{equation*}
is a minimizer to the programming (\ref{eq:Grassl1}).

Note that $\|\x_{\overline{K}}\|_{1}=\frac{n}{2}-k$ and $\|\x-\hat{\x}\|_1=n$,
\begin{eqnarray*}
&&\frac{\|\x-\hat{\x}\|_1}{\|\x_{\overline{K}}\|_{1}}\\
&=&\frac{n}{\frac{n}{2}-k}=\frac{2(\frac{n-k}{k}+1)}{\frac{n-k}{k}-1}\\
&=&\frac{2(C'+1)}{(C'-1)}\\
&>&\frac{2(C+1)}{C-1},
\end{eqnarray*}
strictly contradicting the error bound (\ref{eq:inthmbound}).
\end{proof}

It should be noted that if the condition (\ref{eq:thmeq}) is
true for all the sets $K$ of cardinality $k$, then
\begin{equation*}
2\|\x_{\overline{K}}\|_{1} \geq \frac{(C-1)}{C+1}\|\hat{\x}-\x\|_{1}
\end{equation*}
is also true for the set $K$ which corresponds to the $k$ largest (in amplitude) components of the vector $\x$. So
\begin{equation*}
2\Delta \geq \frac{(C-1)}{C+1}\|\hat{\x}-\x\|_{1}
\end{equation*}
which exactly corresponds to (\ref{eq:noiseadj}). In fact, the
condition (\ref{eq:thmeq}) is also a sufficient and necessary
condition for unique exact recovery of ideally $k$-sparse signals
after we take $C=1$ and let (\ref{eq:thmeq}) take strict inequality
for all $\w\neq 0$ in the null space of $A$.  To see this, suppose
the ideally $k$-sparse signal $\x$ is supported over the set $K$,
namely, $\|\x_{\overline{K}}\|_1=0$. Then from the same triangular
inequality derivation of Theorem \ref{thm:th1}, we know that
$\|\hat{\x}-\x\|_1=0$, namely $\hat{\x}=\x$. Or we can just let $C$
be arbitrarily close to $1$ from the right and since
\begin{equation*}
\|\x-\hat{\x}\|_{1} \leq \frac{2(C+1)}{C-1}  \| \x_{\overline{K}}
\|_{1}=0,
\end{equation*}
we also get $\hat{\x}=\x$. In this sense, when $C=1$, the null space
condition is equivalent to the neighborly polytope condition
\cite{Donoho06} for unique exact recovery of ideally sparse signals.

However, it is an interesting result that, for a particular \emph{fixed} measurement matrix $A$, the violation of (\ref{eq:thmeq}) for some $C>1$ does not necessarily mean that the existence of a vector $\x$ and a minimizer solution $\hat{\x}$ to (\ref{eq:Grassl1}) such that the performance guarantee (\ref{eq:inthmbound}) is violated. For example, assume $n=2$ and the null space of the measurement matrix $A$ is a one-dimensional subspace and has the vector $(1,100)$ as its basis.  Then the null space of the matrix $A$ violates (\ref{eq:thmeq}) with $C=101$ and the set $K=\{1\}$. But a careful examination shows that the biggest possible $\frac{\|\x-\hat{\x}\|_1}{\|\x_{\overline{K}}\|_1}$ ($\|\x_{\overline{K}}\|_1 \neq 0$) is equal to $\frac{100+1}{100}=\frac{101}{100}$, achieved by such an $\x$ as $(-1,-1)$. In fact, all those vectors $\x=(a,b)$ with $b \neq 0$ will achieve $\frac{\|\x-\hat{\x}\|_1}{\|\x_{\overline{K}}\|_1}=\frac{101}{100}$. However, (\ref{eq:inthmbound}) has $\frac{2(C+1)}{C-1}=\frac{204}{100}$. This suggests that for a specific measurement matrix $A$, the tightest error bound for $\frac{\|\x-\hat{\x}\|_1}{\|\x_{\overline{K}}\|_1}$ should involve the detailed structure of the null space of $A$. But for general measurement matrices $A$, as suggested by Theorem \ref{thm:th1}, the condition (\ref{eq:thmeq}) is a necessary and sufficient condition to offer the performance guarantee (\ref{eq:inthmbound}).

It is worth pointing out that the example given in the proof of Theorem \ref{thm:th1} is not just an isolated example. In fact, for two general positive integers $m$ and $n$ with $m < n$ and $n \geq 2$, we can often find an $m \times n$ measurement matrix $A$ and a certain $C>1$ such that the condition (\ref{eq:thmeq}) is violated and, at the same time, for some vector $\x$, the performance bound is also ``tightly" violated.

Consider a generic $m \times n$ matrix  $A'$. For each integer $1\leq k \leq n$, let us define the quantity $h_{k}$ as the supermum of $\frac{\|\w_{K}\|_1}{\|\w_{\overline{K}}\|_1}$ over all such sets $K$ of size $|K|\leq k$ and over all nonzero vectors $\w$ in the null space of $A'$. Let $k^*$ be the biggest $k$ such that $h_{k} \leq 1$. Then there must be a nonzero vector $\w'$ in the null space of $A$ and a set $K^*$ of size $k^*$, such that
\begin{equation*}
\|\w_{K^*}'\|_1 =h_{k^*}\|\w_{\overline{K^*}}'\|_1.
\end{equation*}
Now we generate a new measurement matrix $A$ by multiplying the portion $A'_{K^*}$of the matrix $A'$ by $h_{k^*}$. Then we will have a vector $\w$ in the null space of $A$ satisfying

\begin{equation*}
\|\w_{K^*}\|_1 =\|\w_{\overline{K^*}}\|_1.
\end{equation*}

Now we take a signal vector $\x=(-\w_{K^*},0_{\overline{K^*}})$ and claim that $\hat{\x}=(0,\w_{\overline{K^*}})$ is a minimizer to the programming (\ref{eq:Grassl1}). In fact, recognizing the definition of $h_{k^*}$, we know all the vectors $\w''$ in the null space of the measurement matrix $A$ will satisfy $\|\x+\w''\|_1 \geq \|\x\|_1$. Let us assume that $k^* \geq 2$ and take $K^{''}\subseteq K^*$ as the index set corresponding to the largest $(k^*-i)$ elements of $\x_{K^*}$  in amplitude , where $1 \leq i \leq (k^*-1)$. From the definition of $k^*$, it is apparent that $C'=\frac{\|\w_{\overline{K^{''}}}\|_1}{\|\w_{K^{''}}\|_1}>1$ since $\w$ is nonzero for any index in the set $K^*$.
Let us now take $C=\frac{\|\w_{\overline{K^{''}}}\|_1}{\|\w_{K^{''}}\|_1}+\epsilon$, where $\epsilon>0$ is any arbitrarily small positive number. Thus the condition (\ref{eq:thmeq}) is violated for the vector $\w$, the set $K^{''}$ and the defined constant $C$.

Now by inspection, the decoding error is
\begin{equation*}
\|\x-\hat{\x}\|_1=\frac{2(C'+1)}{C'-1}\|\x_{K^{''}}\|_1>\frac{2(C+1)}{C-1}\|\x_{K^{''}}\|_1,
\end{equation*}
violating the error bound (\ref{eq:inthmbound}) (for the set $K^{''}$).

In the remaining part of this paper, for a given value $\delta=\frac{m}{n}$ and any value $C\geq1$, we will devote our efforts to determining the value of feasible $\zeta=\rho\delta=\frac{k}{n}$ for which there exists a sequence of
$A$ such that the null space condition (\ref{eq:thmeq}) is satisfied for all the sets $K$ of size $k$ when $n$ goes to
infinity and $\frac{m}{n}=\delta$. For a specific $A$, it is very hard to check whether the condition (\ref{eq:thmeq})
is satisfied or not. Instead, we consider randomly choosing $A$ from a Gaussian distribution, and analyze for what $\zeta$, the condition (\ref{eq:thmeq}) for its null space is satisfied with overwhelming
probability as $n$ goes to infinity. When we consider $C=1$,
corresponding to the success of $\ell_{1}$ minimization for all
ideally $k$-sparse signals, loose bounds on $\zeta$ achieving the null space condition were established in \cite{Candes05}\cite{Y}\cite{StXuHassibi08} using the restricted isometry property and high dimensional
geometrical results. The null space condition is equivalent to the $k$-neighborly polytope
condition when $C=1$, so the neighborly polytope condition \cite{Donoho06} gives much sharper bounds for the null space condition when $C=1$. However, no sharp bounds are available for the null space condition with the general case $C \geq 1$.

The standard results on compressive sensing assume that the matrix $A$ has i.i.d. ${\cal N}(0,1)$ entries. The following
lemma gives a characterization of the resulting null space of $A$, which is a fairly well known result, and for the sake of completeness, we include its proof in the appendix.

\begin{lemma}
Let $A\in \mathbb{R}^{m\times n}$ be a random matrix with i.i.d. ${\cal
N}(0,1)$ entries. Then the following statements hold:
\begin{itemize}

\item The distribution of $A$ is right-rotationally invariant: for any $\Theta$ satisfying $\Theta\Theta^*=\Theta^*\Theta=I$, $P_A(A)=P_A(A\Theta)$;
\item There exists a basis $Z$ of the null space of $A$, such that the distribution of $Z$ is left-rotationally invariant:
for any $\Theta$ satisfying $\Theta\Theta^*=\Theta^*\Theta=I$, $P_Z(Z)=P_Z(\Theta^*Z)$;
\item It is always possible to choose a basis $Z$ for the null space such that $Z$ has i.i.d. ${\cal N}(0,1)$ entries.
\label{lemma:nullspaceGaussian}
\end{itemize}
\end{lemma}

In view of Theorem 1 and Lemma 1 what matters is that the null space
of $A$ be rotationally invariantly. Sampling from this rotationally
invariant distribution is equivalent to uniformly sampling a random
$(n-m)$-dimensional subspace from the Grassmann manifold
$\text{Gr}_{(n-m)}(n)$. Here the Grassmann manifold
$\text{Gr}_{(n-m)}(n)$ is the set of $(n-m)$-dimensional subspaces
in the $n$-dimensional Euclidean space $\mathbb{R}^n$ \cite{Grass}. For any
such $A$ and ideally sparse signals, the sharp bounds of
\cite{Donoho06}, apply. However, we shall see that the neighborly
polytope condition for ideally sparse signals does not readily apply to the
proposed null space condition analysis for approximately sparse
signals, since the null space condition can not be transformed to
the $k$-neighborly property in a \emph{single} high-dimensional
polytope \cite{Donoho06}. Instead, in this paper, we shall give a
unified Grassmann angle framework to analyze the proposed null space
property.

\section{The Grassmann Angle Framework for the Null Space Characterization}
\label{sec:Grassprobnull}

In this section we detail the Grassmann angle-based
framework for analyzing the bounds on $\zeta=\frac{k}{n}$ such that
(\ref{eq:thmeq}) holds for every vector in the null space, which we
denote by $Z$. Put more precisely, given a certain constant $C>1$ (or $C \geq 1$), which
corresponds to a certain level of recovery accuracy for the
approximately sparse signals, we are interested in what scaling $\frac{k}{n}$ we can achieve while satisfying the following condition
on $Z$  ($|K|=k$):
\begin{equation}
\forall \w \in Z, \forall K \subseteq \{1,2,...,n\}, C \|\w_{K}\|_1\leq
\|\w_{\overline{K}}\|_1. \label{eq:Grassiffcon}
\end{equation}
From the definition of the condition (\ref{eq:Grassiffcon}), there is a
tradeoff between the largest sparsity level $k$ and the parameter
$C$. As $C$ grows, clearly the largest $k$ satisfying
(\ref{eq:Grassiffcon}) will likely decrease, and, at the same time,
$\ell_{1}$ minimization will be more robust in terms of the the residual norm $\|\x_{\overline{K}}\|_1$. The key in our derivation is the following lemma:
\begin{lemma}
For a certain subset $K \subseteq \{1,2,...,n\}$ with $|K|=k$, the
event that the null space $Z$ satisfies
\begin{equation*}
C \|\w_{K}\|_1\leq \|\w_{\overline{K}}\|_1, \forall \w \in Z
\end{equation*}
is equivalent to the event that $\forall\x$ supported on the
$k$-set $K$ (or supported on a subset of $K$):
\begin{equation}
\|\x_K+\w_{K}\|_1+ \|\frac{\w_{\overline{K}}}{C}\|_1 \geq \|\x_K\|_1,
\forall \w \in Z. \label{eq:Grassxcondition}
\end{equation}
\label{lemma:Grassbaselemma}
\end{lemma}

\begin{proof}
First, let us assume that $C \|\w_{K}\|_1\leq \|\w_{\overline{K}}\|_1,
\forall \w \in Z$.  Using the triangular inequality, we obtain
\begin{eqnarray*}
 &&\| \x_{K}+\w_{K} \|_{1}+\|\frac{\w_{\overline{K}}}{C}\|_{1} \nonumber \\
 & \geq &\|\x_{K}\|_{1}-\|\w_{K}\|_{1}+\|\frac{\w_{\overline{K}}}{C}\|_{1} \nonumber \\
 & \geq &\|\x_{K}\|_{1}.\nonumber
\end{eqnarray*}
thus proving the forward part of this lemma. Now let us assume
instead that $\exists \w \in Z$, such that $C
\|\w_{K}\|_1>\|\w_{\overline{K}}\|_1 $. Then we can construct a vector
$\x$ supported on the set $K$ (or a subset of $K$), with
$\x_{K}=-\w_{K}$. Then we have
\begin{eqnarray*}
 &&\| \x_{K}+\w_{K} \|_{1}+\|\frac{\w_{\overline{K}}}{C}\|_{1} \nonumber \\
 &=& 0+\|\frac{\w_{\overline{K}}}{C}\|_{1} \nonumber \\
 &<&\|\x_{K}\|_{1},\nonumber
\end{eqnarray*}
proving the inverse part of this lemma.
\end{proof}


Now let us consider the probability that condition
(\ref{eq:Grassiffcon}) holds for the sparsity $|K|=k$ if we
uniformly sample a random $(n-m)$-dimensional subspace $Z$ from the
Grassmann manifold $\text{Gr}_{(n-m)}(n)$. Based on Lemma
\ref{lemma:Grassbaselemma}, we can equivalently consider the
complementary probability $P$ that there exists a subset ${K}
\subseteq\{1,2,...,n\}$ with $|K|=k$, and a vector $\x \in \mathbb{R}^n$
supported on the set $K$ (or a subset of $K$) failing the condition
(\ref{eq:Grassxcondition}). With the linearity of the subspace $Z$
in mind, to obtain $P$, we can restrict our attention to those
vectors $\x$ from the cross-polytope (the unit $\ell_1$ ball)
\begin{equation*}
\{\x\in \mathbb{R}^n ~|~\|\x\|_1=1 \}
\end{equation*}
that are only supported on the set $K$ (or a subset of $K$).

 First, we upper bound the probability $P$ by a union bound over all the possible support sets $K \subseteq \{1,2,...,n\}$ and
 all the sign patterns of the $k$-sparse vector
 $\x$. Since the $k$-sparse vector $\x$ has $\binom{n}{k}$ possible support sets of cardinality $k$ and $2^k$ possible sign
 patterns (nonnegative or nonpositive), we have

 \begin{equation}
  P\leq \binom{n}{k} \times 2^k \times P_{K,-},
 \label{eq:union}
 \end{equation}
where $P_{K,-}$ is the probability that for a specific \emph{support
set} $K$, there exist a $k$-sparse vector $\x$ of a specific
\emph{sign pattern} which  fails the condition
(\ref{eq:Grassxcondition}). By symmetry, without loss of generality,
we assume the signs of the elements of $\x$ to be nonpositive.

So now let us focus on deriving the probability $P_{K,-}$.
Since $\x$ is a nonpositive $k$-sparse vector supported on the set
$K$ (or a subset of $K$) and can be restricted to the cross-polytope
$\{\x \in \mathbb{R}^n~|~\|\x\|_1=1 \}$, $\x$ is also on a
$(k-1)$-dimensional face, denoted by $F$, of the skewed
cross-polytope (weighted $\ell_1$ ball) SP:
\begin{equation}
\text{SP}=\{\y\in \mathbb{R}^n~|~\|\y_K\|_1+ \| \frac{\y_{\overline{K}}}{C}\|_1
\leq 1\}
\end{equation}


\begin{figure*}
\centering
\includegraphics[width=5.4 in, height=4.4 in]{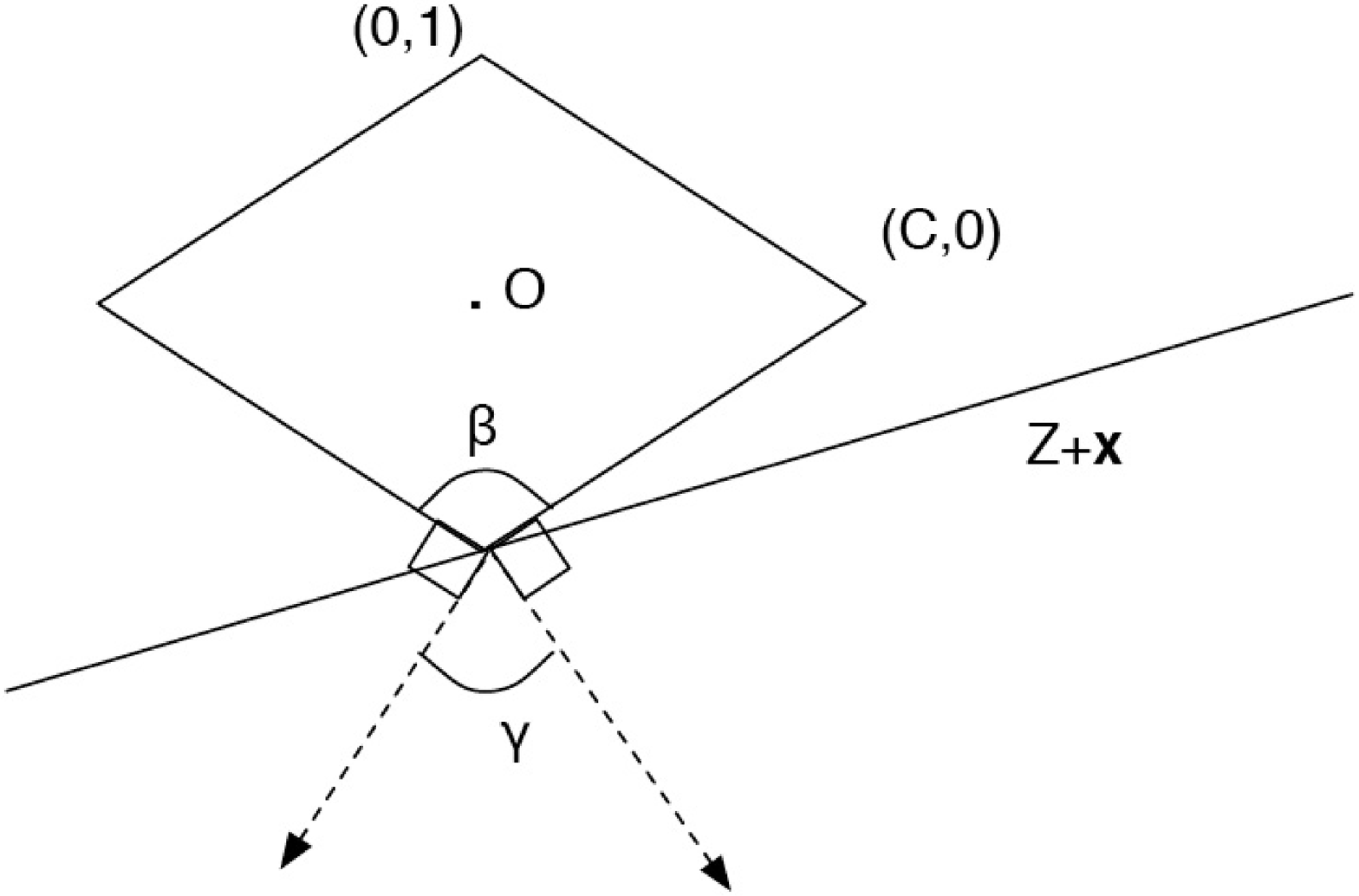}
\caption{The Grassmann Angle for a Skewed Cross-polytope}
\label{fig:grassangle}
\end{figure*}

Then $P_{K,-}$ is the probability that there exists an $\x \in F$,
and there exists a $ \w \in Z$ ($\w\neq 0$) such that
\begin{equation} \|\x_K+\w_{K}\|_1+ \|\frac{\w_{\overline{K}}}{C}\|_1 \leq
\|\x_K\|_1=1.
\label{eq:perturbationsmaller}
\end{equation}
We first focus on studying a specific \emph{single }point $\x
\in F$, without loss of generality, assumed to be in the relative
interior of this $(k-1)$ dimensional face $F$. For this single particular
$\x$ on the $F$, the probability, denoted by $P_{\x}'$, that $\exists \w
\in Z$ ($\w\neq 0$) such that \ref{eq:perturbationsmaller} holds
is essentially the probability that a uniformly chosen $(n-m)$
dimensional subspace $Z$ shifted by the point $\x$, namely $(Z+\x)$,
intersects the skewed cross-polytope
\begin{equation}
\text{SP}=\{\y\in \mathbb{R}^n~|~\|\y_K\|_1+ \|\frac{\y_{\overline{K}}}{C}\|_1
\leq 1\}
\end{equation}
\emph{nontrivially}, namely, at some other point besides $\x$.

From the linear property of the subspace $Z$, the event that
$(Z+\x)$ intersects the skewed cross-polytope SP is equivalent to the
event that $Z$ intersects nontrivially with the cone SP-Cone($\x$)
obtained by observing the skewed polytope SP from the point $\x$.
(Namely, SP-Cone($\x$) is conic hull of the point set
$(\text{SP}-\x)$ and SP-Cone($\x$) has the origin of the
coordinate system as its apex.) However, as noticed in the geometry
for convex polytopes \cite{Grunbaum68}\cite{Grunbaumbook}, the
SP-Cone(\x) are identical for any $\x$ lying in the relative
interior of the face $F$.  This means that the probability $P_{K,-}$
is equal to $P_{\x}'$, regardless of the fact $\x$ is only a single
point in the relative interior of the face $F$. (The acute reader
may have noticed some singularities here because $\x \in F$ may not
be in the relative interior of $F$, but it turns out that the
SP-Cone(\x) is then only a subset of the cone we get when $\x$ is in
the relative interior of $F$. So we do not lose anything if we
restrict $\x$ to be in the relative interior of the face $F$.) In
summary, we have
\begin{equation*}
P_{K,-}=P_{\x}'.
\end{equation*}

Now we only need to determine $P_{\x}'$. From its definition,
$P_{\x}'$ is exactly the \emph{\textbf{complementary Grassmann
angle}} \cite{Grunbaum68} for the face $F$ with respect to the
polytope SP under the Grassmann manifold
$\text{Gr}_{(n-m)}(n)$:\footnote{A Grassman angle and its
corresponding complementary Grassmann angle always sum up to 1.
There is apparently inconsistency in terms of the definition of
which is ``Grassmann angle" and which is ``complementary Grassmann
angle" between \cite{Grunbaum68},\cite{Schneider} and
\cite{Vershik92} etc. But we will stick to the earliest definition
in \cite{Grunbaum68} for Grassmann angle: the measure of the
subspaces that intersect trivially with a cone. } the probability of
a uniformly distributed $(n-m)$-dimensional subspace $Z$ from the
Grassmannian manifold $\text{Gr}_{(n-m)}(n)$ intersecting
nontrivially with the cone SP-Cone($\x$) formed by observing the
skewed cross-polytope SP from the relative interior point $\x \in
F$.

Building on the works by L.A.Santal\"{o} \cite{Santalo1952} and
P.McMullen \cite{McMullen1975} etc. in high dimensional integral
geometry and convex polytopes, the complementary Grassmann angle for
the $(k-1)$-dimensional face $F$ can be explicitly expressed as the
sum of products of internal angles and external angles
\cite{Grunbaumbook}:
\begin{equation}
2\times \sum_{s \geq 0}\sum_{G \in \Im_{m+1+2s}(\text{SP})}
{\beta(F,G)\gamma(G,\text{SP})}, \label{eq:Grassangformula}
\end{equation}
where $s$ is any nonnegative integer, $G$ is any
$(m+1+2s)$-dimensional face of the skewed cross-polytope
($\Im_{m+1+2s}(\text{SP})$ is the set of all such faces),
$\beta(\cdot,\cdot)$ stands for the internal angle and
$\gamma(\cdot,\cdot)$ stands for the external angle.

The internal angles and external angles are basically defined as
follows \cite{Grunbaumbook}\cite{McMullen1975}:
\begin{itemize}
\item An internal angle $\beta(F_1, F_2)$ is the fraction of the
hypersphere $S$ covered by the cone obtained by observing the face
$F_2$ from the face $F_1$. \footnote{Note the dimension of the
hypersphere $S$ here matches the dimension of the corresponding cone
discussed. Also, the center of the hypersphere is the apex of the
corresponding cone. All these defaults also apply to the definition
of the external angles. } The internal angle $\beta(F_1, F_2)$ is
defined to be zero when $F_1 \nsubseteq F_2$ and is defined to be
one if $F_1=F_2$.
\item An external angle $\gamma(F_3, F_4)$ is the fraction of the
hypersphere $S$ covered by the cone of outward normals to the
hyperplanes supporting the face $F_4$ at the face $F_3$. The
external angle $\gamma(F_3, F_4)$ is defined to be zero when $F_3
\nsubseteq F_4$ and is defined to be one if $F_3=F_4$.

\end{itemize}

Let us take for example the $2$-dimensional skewed cross-polytope
\begin{equation*}
\text{SP}=\{(y_1, y_2)\in \mathbb{R}^2|~\|\y_2\|_1+ \|\frac{\y_1}{C}\|_1 \leq
1\}
\end{equation*}
(namely the diamond) in Figure \ref{fig:grassangle}, where $n$=2,
$(n-m)=1$ and $k=1$. Then the point $\x=(0,-1)$ is a 0-dimensional
face (namely a vertex) of the skewed polytope SP. Now from their
definitions, the internal angle $\beta(\x, \text{SP})=\frac{\beta}{2\pi}$ and the
external angle $\gamma(\x,\text{SP})=\frac{\gamma}{2\pi}$,
$\gamma(\text{SP},\text{SP})=1$. The complementary Grassmann angle
for the vertex $\x$ with respect to the polytope SP is the
probability that a uniformly sampled $1$-dimensional subspace (namely
a line, we denote it by $Z$) shifted by $\x$ intersects
nontrivially with $\text{SP}=\{(y_1, y_2)\in \mathbb{R}^2|~\|\y_2\|_1+
\|\frac{\y_1}{C}\|_1 \leq 1\}$ (or equivalently the probability that
$Z$ intersects nontrivially with the cone obtained by observing SP
from the point $\x$). It is obvious that this probability is
$\frac{\beta}{\pi}$. The readers can also verify the correctness of the formula
(\ref{eq:Grassangformula}) very easily for this toy example.

Generally, it might be hard to give explicit formulae for the
external and internal angles involved, but fortunately in the skewed
cross-polytope case, both the internal angles and the external
angles can be explicitly computed.

Firstly, let us look at the internal angle $\beta(F,G)$ between the
$(k-1)$-dimensional face $F$ and a $(l-1)$-dimensional face $G$.
Notice that the only interesting case is when $F \subseteq G$ since
$\beta(F,G)\neq 0$ only if $F \subseteq G$. We will see if $F
\subseteq G$, the cone formed by observing $G$ from $F$ is the
direct sum of a $(k-1)$-dimensional linear subspace and a convex
polyhedral cone formed by $(l-k)$ unit vectors with inner product
$\frac{1}{1+C^2k}$ between each other. In this case, the internal
angle is given by
\begin{equation}
\beta(F,G)=\frac{V_{l-k-1}(\frac{1}{1+C^2k},l-k-1)}{V_{l-k-1}(S^{l-k-1})},
\label{eq:internal}
\end{equation}
where $V_i(S^i)$ denotes the $i$-th dimensional surface measure on
the unit sphere $S^{i}$, while $V_{i}(\alpha', i)$ denotes the
surface measure for regular spherical simplex with $(i+1)$ vertices
on the unit sphere $S^{i}$ and with inner product as $\alpha'$
between these $(i+1)$ vertices. Thus (\ref{eq:internal}) is equal to
$B(\frac{1}{1+C^2k}, l-k)$, where
\begin{equation}
B(\alpha', m')=\theta^{\frac{m'-1}{2}} \sqrt{(m'-1)\alpha' +1}
\pi^{-m'/2} {\alpha'}^{-1/2}J(m',\theta),
\end{equation}
with $\theta=(1-\alpha')/\alpha'$ and
\begin{equation}
 J(m', \theta)=\frac{1}{\sqrt{\pi}} \int_{-\infty}^{\infty}(\int_{0}^{\infty} e^{-\theta v^2+2i v\lambda} \,dv )^{m'} e^{-\lambda^2} \,d\lambda.
\end{equation}

We should remark that the formula above for the internal angle is true only when the face $G$ is not of dimension $n$. When $G$ is $n$-dimensional, we will derive a separate formula in Lemma \ref{lemma:internalSP}. Since the expression for this special case will not affect our following derivations in a significant way, we choose not to list it here.

Secondly, we can derive the external angle $\gamma(G, \text{SP})$
between the $(l-1)$-dimensional face $G$ and the skewed
cross-polytope SP as:
\begin{equation}
\gamma(G, \text{SP})=\frac{2^{n-l}}{{\sqrt{\pi}}^{n-l+1}}
\int_{0}^{\infty}e^{-x^2}(\int_{0}^{\frac{x}{C\sqrt{k+\frac{l-k}{C^2}}}}
e^{-y^2} \,dy)^{n-l}\,dx. \label{eq:cncsex}
\end{equation}
The derivations of these expressions involve the computations of the
volumes of cones in high dimensional geometry and will be presented
in the appendix.

In summary, combining (\ref{eq:union}), (\ref{eq:Grassangformula}),
(\ref{eq:internal}) and (\ref{eq:cncsex}), we get an upper bound on
the probability $P$. If we can show that for a certain
$\zeta=\frac{k}{n}$, $P$ goes to zero exponentially in $n$ as
$n\rightarrow \infty$, then we know that for such $\zeta$, the
null space condition (\ref{eq:Grassiffcon}) holds with overwhelming
probability. This is the guideline for computing the bound on
$\zeta$ in the following sections.

\section{Evaluating the Bound $\zeta$}
\label{sec:evathebound}

In summary,
 \begin{equation}
  P\leq \binom{n}{k} \times 2^k \times 2\times \sum_{s \geq 0}\sum_{G \in \Im_{m+1+2s}(\text{SP})}
{\beta(F,G)\gamma(G,\text{SP})}. \label{eq:num4}
\end{equation}

In order for this upper bound on $P$ to decrease to 0 as
$n\rightarrow \infty$, one sufficient condition is that every sum term in
(\ref{eq:num4}) goes to $0$ exponentially fast in $n$. We remark that the equation in (\ref{eq:num4}) is similar to the
expected number of missed ``faces" in the study of $k$-neighborly
polytope \cite{Donoho06, Vershik92}, but generalizes the
$k$-neighborly polytope formula to more general Grassmann angles. In
the following sections, we will extend the techniques developed in
\cite{Donoho06, Vershik92} to evaluating the bounds on $\zeta$ from
(\ref{eq:num4}), taking into account of the variable $C>1$. To illustrate the effect of $C$ on the bound $\zeta$, also for the sake of completeness, we will keep the detailed derivations.

For simplicity of analysis, we define $l=(m+1+2s)+1$ and $\nu=\frac{l}{n}$. In the skewed cross-polytope SP, we notice that
there are in total $\binom{n-k}{l-k} 2^{l-k}$ faces $G$ of dimension
$(l-1)$ such that $F\subseteq G$ and $\beta(F, G) \neq 0$. Because of the symmetry in
the skewed cross-polytope SP, it follows from (\ref{eq:num4}) that

 \begin{equation}
  P\leq \sum_{s \geq 0}\underbrace{{\underbrace{2\binom{n}{k}2^l \times\binom{n-k}{l-k}}_{COM_{s}}\beta(F,G)\gamma(G,\text{SP})}}_{D_{s}},
\label{eq:num3}
\end{equation}
where $l=(m+1+2s)+1$ and $G\subseteq \text{SP}$ is any single face
of dimension $(l-1)$ such that $F \subseteq G$.

Closely following the approach of \cite{Donoho06}, in estimating $n^{-1}\log(D_{s})$, we
can decompose it into a sum of terms involving logarithms of the
combinatorial factor, the internal angle and the external angle.
With
\begin{equation*}
H(p) = p \log(1/p) + (1-p)\log(1/(1-p)),
\end{equation*}
where the logarithm base is over $e$. From Stirling's formula, we know that
\begin{equation}
n^{-1}\log \binom{n}{\lfloor pn \rfloor}\rightarrow H(p), p \in [0,
1], n\rightarrow \infty. \label{eq: entrasy}
\end{equation}

Defining $\nu = l/n \geq \delta$, we have
\begin{equation}
n^{-1} \log(COM_{s}) = \nu \log(2) + H(\rho\delta) +
H(\frac{\nu-\rho\delta}{1-\rho\delta} )(1-\rho\delta) + R_1
\label{eq:comexp}
\end{equation}
with remainder $R_1 = R_1(s, k, m, n)$.

Define the combinatorial growth exponent for $COM_{s}$
\begin{equation}
\psi_{com}(\nu;\rho, \delta) = \nu\log(2) + H(\rho\delta) + H(
\frac{\nu-\rho\delta}{1-\rho\delta} )(1-\rho \delta),
\end{equation}
describing the exponential growth of the combinatorial factors.
Applying (\ref{eq: entrasy}), we will see that the remainder $R_1$
in (\ref{eq:comexp}) is $o(1)$ uniformly in the range $l>\delta n$,
$n>n_0(\rho, \delta, \epsilon)$, where $n_0(\rho, \delta, \epsilon)$ is some big enough natural number.

For a particular $C$, we will also define a decay exponent $\psi_{ext}(\nu;\rho,\delta)$ and show
that $\gamma(G,\text{SP})$ decays exponentially at least at the rate $\psi_{ext}{(\nu;\rho,\delta)}$: for each $\epsilon>0$,
\begin{equation*}
n^{-1}\log(\gamma(G, \text{SP}))\leq -\psi_{ext}(\nu)+\epsilon,
\end{equation*}
uniformly in $l \geq \delta n$, $n \geq n_{0}(\rho, \delta, \epsilon)$. When it is clear in the context what $C$ is, we will often omit $C$ in the notations.

Similarly, under the parameter $C$, Section \ref{sec:bndinternal} below shows that the
decay exponent for the internal angle $\beta(F,G)$ is
$\psi_{int}(\nu; \rho,\delta)$, which is defined in Section \ref{sec:bndinternal}. Since $k \sim \rho \delta n$, $l \sim \nu n$, we will have the scaling
\begin{equation*}
n^{-1} \log(\beta(F,G))=-\psi_{int}(\nu;\rho,\delta)+R_2,
\end{equation*}
where the remainder $R_2=o(1)$ uniformly in $l \geq \delta n $ when $n \geq n_{0}(\rho, \delta, \epsilon)$ is a large enough natural number.

In summary, under a given $C>1$, for any fixed choice of $\rho$, $\delta$, for
$\epsilon > 0$, and for $n \geq n_0(\rho, \delta, \epsilon)$,
\begin{equation}
n^{-1} \log(Ds)\leq \psi_{com}(\nu;\rho,
\delta)-\psi_{int}(\nu;\rho, \delta)-\psi_{ext}(\nu;\rho,
\delta)+3\epsilon , \label{eq:comexpdef}
\end{equation}
holds uniformly over the sum parameter $s$ in (\ref{eq:Grassangformula}).

In the rest of this paper, when the parameters $\rho$, $\delta$ and $C$ are clear from the context, we will omit them  from the notations for the combinatorial, internal and external exponents.

\subsection{Characterizing $\rho_{N}(\delta, C)$}
Continuing to follow \cite{Donoho06}, we define the \emph{net exponent} $\psi_{net}=\psi_{com}(\nu;\rho,
\delta)-\psi_{int}(\nu;\rho, \delta)-\psi_{ext}(\nu;\rho, \delta)$.
We will know that the components of $\psi_{net}$ are all continuous
over sets $\rho\in[\rho_0,1], \delta \in [\delta_0,1], \nu\in
[\delta,1]$, and $\psi_{net}$ is also continuous over these regions.

\begin{definition}
Let $\delta \in (0,1]$. The critical proportion $\rho_N(\delta, C)$
is the supremum of $\rho \in [0,1]$ satisfying
\begin{equation*}
 \psi_{net} (\nu; \rho, \delta)<0, ~~~~\nu \in [\delta, 1].
\end{equation*}
Continuity of $\psi_{net}$ shows that if $\rho<\rho_{N}$ then, for
some $\epsilon>0$,
\begin{equation*}
 \psi_{net} (\nu; \rho, \delta)<-\epsilon, ~~~~\nu \in [\delta, 1].
\end{equation*}

Combine this with (\ref{eq:comexpdef}), for all $s=0,2, \ldots
,(n-m)/2$ and all $n>n_0(\delta, \rho,\epsilon)$,
\begin{equation*}
n^{-1}\log(D_{s}) \leq -\epsilon.
\end{equation*}
Note that if this negative exponent condition holds, we will have the results in Theorem \ref{thm:th1}.

\end{definition}

In the next section, we will specify the exponents $\psi_{int}$ and
$\psi_{ext}$ for the internal angles and external angles respectively,  and we will discuss properties of $\rho_N (\delta, C)$.

\section{Characterizations of Angle Exponents}
\label{sec:porexp}

\subsection{Exponent for External Angle}
Let $G$ denote the cumulative distribution function of a half-normal
$HN(0, 1/2)$ random variable, i.e. a random variable $X = |Z|$ where
$Z \sim N(0, 1/2)$, and $G(x) = Prob\{X \leq x\}$, where the density
function $g(x) = 2/\sqrt{\pi}\exp(-x^2)$ and thus $G(x)$ is the
error function
\begin{equation}
G(x) = \frac{2}{\sqrt{\pi}} \int_{0}^{x}e^{-y^2}\,dy.
\label{eq:Grasserf}
\end{equation}
For $\nu \in (0, 1]$, define $x_{\nu}$ as the solution of
\begin{equation}
\frac{2xG(x)}{g(x)}=\frac{1-\nu}{\nu'}, \label{eq:external}
\end{equation}
where
\begin{equation*}
\nu'=(C^2-1)\rho\delta+\nu.
\end{equation*}

Because $xG(x)$ is a smooth strictly increasing function, which goes to
$0$ as $x \rightarrow 0$ and behaves close to $x$ as $x\rightarrow
\infty$, and because $g(x)$ is strictly decreasing, the function
$2xG(x)/g(x)$ is a strictly increasing function. So $x_{\nu}$ is a well-defined,
smooth, and decreasing function of $\nu$.

We have $x_{\nu} \rightarrow 0$ as $\nu \rightarrow 1$ and $x_{\nu}
\sim \sqrt{\log((1-\nu)/\nu')}$ as $\nu \rightarrow 0$. Define now
\begin{equation*}
\psi_{ext}(\nu) = -(1-\nu) \log(G(x_{\nu})) + \nu x_{\nu}^2 .
\end{equation*}
This function is smooth on the interior of (0, 1), with endpoints
$\psi_ {ext}(1) = 0$, $\psi_{ext}(0) = 0$. When $C=1$, we have the
asymptotic \cite{Donoho06}
\begin{equation}
\psi_{ext}(\nu) \sim \nu \log(\frac{1}{\nu} )-\frac{1}{2}\nu
\log(\log(\frac{1}{\nu})) + o(\nu), \nu\rightarrow 0.
\label{eq:extasyp}
\end{equation}

\subsection{Exponent for Internal Angle}
Closely following \cite{Donoho06}, take $Y$ as a standard half-normal random variable $HN(0, 1)$. From standard calculations, we know that its cumulant generating function $\Lambda(s)=\log(E(\exp(sY))$ is given by
\begin{equation*}
\Lambda(s)=\frac{s^2}{2}+\log(2\Phi(s)),
\end{equation*}
where $\Phi$ is the usual cumulative distribution function of a
standard Normal $N(0,1)$. So the large deviation rate function
of the cumulant generating function $\Lambda^*$ is defined as
\begin{equation*}
\Lambda^*(y)={\max_{s}}~sy-\Lambda(y).
\end{equation*}
From the large deviation theory, this function is smooth and convex on $(0,\infty)$, strictly positive except being equal to $0$
at $\mu=E(Y)=\sqrt{2/\pi}$. For $\gamma' \in (0,1)$ let
\begin{equation}
\xi_{\gamma'}(y)=\frac{1-\gamma'}{\gamma'} y^2/2+\Lambda^*(y),
\label{eq:intexpfun}
\end{equation}
where we define
\begin{equation*}
\gamma'=\frac{\rho\delta}{\frac{C^2-1}{C^2}\rho\delta+\frac{\nu}{C^2}}.
\end{equation*}

The function $\xi_{\gamma'}(y)$ is strictly convex
and positive on $(0,\infty)$ and has a unique minimum at a unique
$y_{\gamma'}$ in the interval $(0, \sqrt{2/\pi})$. Then we have the
internal angle exponent as
\begin{equation}
\psi_{int}(\nu;\rho,\delta)=\xi_{\gamma'}(y_{\gamma'})(\nu-\rho\delta)+\log(2)(\nu-\rho\delta).
\label{eq:inteptfor}
\end{equation}
For fixed $\rho$, $\delta$, $\Lambda_{int}$ is continuous in $\nu
\geq \delta$. Most importantly, in the section below, we get the
asymptotic formula
\begin{equation}
\xi_{\gamma'}(y_{\gamma'})\sim
\frac{1}{2}\log(\frac{1-\gamma'}{\gamma'}), \gamma' \rightarrow 0
\label{eq:intasyp1}
\end{equation}
Because
$\gamma'=\frac{\rho\delta}{\frac{C^2-1}{C^2}\rho\delta+\frac{\nu}{C^2}}$,
(\ref{eq:intasyp1}) means for small $\rho$, $\nu\in [\delta, 1]$ and
any given $\eta>0$
\begin{equation}
\psi_{int}(\nu, \rho\delta)\geq (\frac{1}{2}\cdot
\log(\frac{1-\gamma'}{\gamma'})(1-\eta)+\log(2))(\nu-\rho\delta).
\label{eq:intsmallasp}
\end{equation}

\subsection{Properties of $\rho_{N}(\delta, C)$}
We now consider the combined behavior of $\psi_{com}$, $\psi_{int}$
and $\psi_{net}$. We think of these as functions of $\nu$ with
$\rho$, $\delta$ as parameters. 
$\psi_{com}$ is the exponent of a growing function which must be
outweighed by the sum of the other two exponents:
$\psi_{int}+\psi_{net}$.

The asymptotic relations (\ref{eq:extasyp}) and (\ref{eq:intasyp1})
allow us to see the following key facts about $\rho_{N}(\delta, C)$,
the proofs of which are given in the appendix.
\begin{lemma}
For any $\delta>0$ and any $C>1$, we have
\begin{equation}
\rho_{N}(\delta, C)>0, \delta \in (0,1).
\end{equation}
\label{lemma:existence}
\end{lemma}

Generalizing the result in \cite{Donoho06} for $\rho_N (\delta, 1)$, one can show the asymptotic of $\rho_N (\delta, C)
\rightarrow 0$ as $\delta \rightarrow 0$.
\begin{lemma}
For all $\eta>0$ and any $C> 1$,
\begin{equation}
\rho_{N}(\delta,C) \geq \log(\frac{1}{\delta})^{-(1+\eta)},
~~~\delta \rightarrow 0.
\end{equation}
\label{lemma:netasyp}
\end{lemma}

Finally, we have the lower and upper bounds for $\rho_{N}(\delta,
C)$, which shows the scaling bounds for $\rho_{N}(\delta,
C)$ as a function of $C$.
\begin{lemma}
When $C\geq 1$, for any fixed  $\delta>0$,
\begin{equation}
\Omega(\frac{1}{C^2})\leq \rho_{N}(\delta,C) \leq \frac{1}{C+1},
\end{equation}
where $\Omega(\frac{1}{C^2})\leq \rho_{N}(\delta, C)$ means that
there exists a constant $\iota (\delta)$,
\begin{equation*}
\frac{\iota (\delta)}{C^2} \leq \rho_{N}(\delta, C), ~~~\text{as}~~C \rightarrow \infty,
\end{equation*}
where we can take $\iota(\delta)=\rho_{N}(\delta, 1)$.
\label{lemma:Csyp}
\end{lemma}

\section{Deriving the External Angle Exponents}
\label{sec:bndexnangle}

In the previous section, we described how to compute the external and internal angle exponents, and we will give the derivations which justify the computations of the two exponents. First, we start justifying the computation of $\psi_{ext}$ given in Section \ref{sec:porexp}.
\begin{lemma}
Fix $\delta$, $\epsilon>0$
  \begin{equation}
  n^{-1} \log(\gamma(G, \text{SP}))<- \psi_{net}(l/n)+\epsilon_{1},
  \end{equation}
uniformly in $l \geq \delta n$, when $n$ is large enough.
\label{lemma:Grassextasy}
\end{lemma}

\begin{proof}
In the appendix, we derived the explicit integral formula for
the external angle:
\begin{equation}
\gamma(G, \text{SP})=\frac{2^{n-l}}{{\sqrt{\pi}}^{n-l+1}}
\int_{0}^{\infty}e^{-x^2}(\int_{0}^{\frac{x}{C\sqrt{k+\frac{l-k}{C^2}}}}
e^{-y^2} \,dy)^{n-l}\,dx.
\end{equation}

After a changing of integral variables, we have
\begin{eqnarray}
&&\gamma(G, \text{SP})=\sqrt{\frac{(C^2-1)k+l}{{\pi}}}\\ \nonumber
&&\int_{0}^{\infty}e^{-((C^2-1)k+l)x^2}(\frac{2}{\sqrt{\pi}}
\int_{0}^{x} e^{-y^2} \,dy)^{n-l}\,dx.
\end{eqnarray}

Inside the parenthesis is the error function $G$ from
(\ref{eq:Grasserf}). Let $\nu=\frac{l}{n}$, $\nu'=(C^2-1)\rho \delta +\nu$
then the integral formula can be written as
\begin{equation}
\sqrt{\frac{n \nu'}{{\pi}}} \int_{0}^{\infty}e^{-n \nu'
x^2+n(1-\nu)\log(G(x))} \,dx. \label{eq:fuun}
\end{equation}

To look at the asymptotic behavior of (\ref{eq:fuun}), following the same methodology as in \cite{Donoho06}, we use Laplace's method. We define
\begin{equation}
f_{\rho,\delta, \nu,n}(y)=e^{-n \psi_{\rho,\delta, \nu}(y)} \cdot
\sqrt{\frac{n \nu'}{{\pi}}} \label{eq:Grassfnnu}
\end{equation}
with
\begin{equation*}
\psi_{\rho,\delta, \nu}(y)=\nu' y^2-(1-\nu)\log(G(y))
\end{equation*}

We will develop expressions for the second and third derivatives of the function $\psi_{\rho,\delta, \nu}$. Applying Laplace's method to $\psi_{\rho,\delta, \nu}$ gives the following lemma, where we will defer the proof to later parts of the paper.

\begin{lemma}
For $\nu \in (0,1)$, let $x_{\nu}$ denote the minimizer of
$\psi_{\rho,\delta, \nu}$. Then
\begin{equation*}
\int_{0}^{\infty} f_{\rho,\delta, \nu,n}(x) \,dx \leq
e^{-n\psi_{\rho,\delta, \nu}(x_{\nu})(1+R_n(\nu))},
\end{equation*}
where for $\delta, \eta>0$,
\begin{equation*}
\sup_{\nu \in [\delta, 1-\eta]} R_n(\nu)=o(1)~~as~~n\rightarrow
\infty,
\end{equation*}
and $x_{\nu}$ is exactly the same $x_{\nu}$ defined
earlier in (\ref{eq:external}).
\label{lemma:Grassextlap}
\end{lemma}

Recall that the defined exponent $\psi_{ext}$ is given by
\begin{equation}
\psi_{ext}(\nu;\rho, \delta)=\psi_{\rho, \delta, \nu}(x_{\nu}).
\label{eq:Grassbigsmall}
\end{equation}

Using the definition of $\psi_{\rho, \delta, \nu}(x_{\nu})$ and (\ref{eq:Grassbigsmall}), it is not hard to see, as $\nu \rightarrow 1$,
$x_{\nu} \rightarrow 0$ and $\psi_{ext}(\nu) \rightarrow 0$.
For any given $\epsilon_{1} >0$ in Lemma \ref{lemma:Grassextasy}, there is a largest $\nu_{\epsilon_{1}}<1$ with
$\psi_{ext}(\nu_{\epsilon_{1}}) =\epsilon_{1}$. Note that $\gamma(G,\text{SP}) \leq
1$, so for $l>\nu_{\epsilon_{1}}n$,
\begin{equation*}
n^{-1}\log(\gamma(G,\text{SP}))\leq 0< -\psi_{ext}(\nu)+\epsilon_{1},
\end{equation*}
for $n \geq 1$. Consider now $l \in [\delta n, \nu_{\epsilon_{1}} n]$,
based on (\ref{eq:fuun}),
\begin{equation*}
\gamma(G,\text{SP})=\int_{0}^{\infty} f_{\rho,\delta, \nu,n}(y)
\,dx.
\end{equation*}
From Lemma \ref{lemma:Grassextlap}, as $n \rightarrow \infty$, uniformly for $l \in [\delta n, \nu_{\epsilon_{1}} n]$,
\begin{equation*}
n^{-1} \log(\gamma(G,\text{SP}))=\psi_{\nu}(x_{\nu})+o(1),
\end{equation*}
where we have abbreviated $\psi_{\rho, \delta, \nu}(\cdot)$ to
$\psi_{\nu}(\cdot)$ for fixed $\rho$ and $\delta$.

So from the identity (\ref{eq:Grassbigsmall}), we get
\begin{equation}
n^{-1}\log(\gamma(G,\text{SP})) \leq -\psi_{net}(l/n)+o(1).
\end{equation}
Then Lemma \ref{lemma:Grassextasy} follows.

\end{proof}

Now it remains to prove the uniformity result for Laplace's method in Lemma
\ref{lemma:Grassextlap}. We will follow the same line of reasoning
given in \cite{Donoho06}. First, we state explicitly the key lemma
from \cite{Donoho06}.

\begin{lemma}
\cite{Donoho06} Let $\psi(x)$ be convex in $x$ and belong to the
differentiability class $\mathbb{C}^2$ (the second derivative exists and is
continuous) on an interval $I$ and suppose that it takes its minimum
at an interior point $x_0 \in I$, where $\psi''(x_{0})>0$ and that in a
vicinity $(x_0-\epsilon, x_0+\epsilon)$ of $x_0$:
\begin{equation}
|\psi''(x)-\psi''(x_0)|\leq D|\psi''(x_0)||x-x_0|.
\end{equation}

Let $\overline{\psi}$ be the quadratic approximation
$\psi(x_0)+\psi''(x_0)(x-x_0)^2/2$. Then
\begin{equation*}
\int_{I}{\exp(-n\psi(x))}\,dx \leq \int_{-\infty}^{\infty} \exp{(-n
\overline{\psi}(x))} \,dx\cdot(S_{1,n}+S_{2,n})
\end{equation*}
where
\begin{equation*}
S_{1,n}=\exp(n\psi''(x_0)D\epsilon^3/6)
\end{equation*}
\begin{equation*}
S_{2,n}=2/\left(n\epsilon(2\pi|\psi''(0)|)^{\frac{1}{2}}(1-\frac{1}{2}D\epsilon^2)\right)
\end{equation*}
\label{lemma:uniform}
\end{lemma}

The constant $D$ in this lemma can be a scaled third derivative,
since if $\psi$ is $\mathbb{C}^3$, we can take
\begin{equation*}
D=\sup_{(x_0-\epsilon,x_0+\epsilon)} \psi^{(3)}(x)/ \psi''(x).
\end{equation*}

Based on Lemma \ref{lemma:uniform}, we can derive the uniformity in
Lemma \ref{lemma:Grassextlap}. In fact, if we pick
$\epsilon_{n}=n^{-\frac{2}{5}}$ and let $n \geq n_1(\psi''(x_0),
D)$, where $n_1(\psi''(x_0),D)$ is a number depending only on
$\psi''(x_0)$ and $D$, we can use
\begin{equation}
\int_{I} e^{-n \psi(x)}\,dx \leq  \int_{-\infty}^{\infty} e^{-n
\overline{\psi}(x)} \,dx \cdot (1+o(1))
\end{equation}

Here the term $o(1)$  is uniform over any collection of convex
functions with a given $\psi''(x_0)$ and $D$. From here to the end of this section, we will abbreviate $\psi_{\rho, \delta,\nu}$ as $\psi_{\nu}$ for the fixed parameter $\rho$ and $\delta$.

Now we consider the collection of convex functions $\psi_{\nu}~~(\nu \in [\delta,
1-\eta])~~$in Lemma \ref{lemma:Grassextlap}. Following the
derivations in \cite{Donoho06}, if we can show that there exists a
certain $\epsilon>0$ so that $\psi''({x_{0}})$ and $D$ is bounded for
the function $\psi_{\nu}(x)$ uniformly over the range $\nu \in
[\delta, 1-\eta]$, then Lemma \ref{lemma:Grassextlap} holds. Indeed, this is true based on Lemma \ref{lemma:derivn} as given below.

\begin{lemma}
The function $\psi_{\nu}(\cdot)$ is smooth with its second
derivative at $x_{\nu}$
\begin{equation}
\psi_{\nu}''(x_\nu)=2\nu'+4x_{\nu}^2\nu'+\frac{4x_{\nu}^2\nu'^2}{1-\nu}
\label{eq:2nd}
\end{equation}
and its third derivative at $x_{\nu}$
\begin{equation}
\psi_{\nu}^{(3)}(x_{\nu})=(1-\nu)\left((2-4x_{\nu}^2)z-6x_{\nu}z^2-2z^3\right)
\label{eq:3rd}
\end{equation}
where $z=z_{\nu}=2\nu' x_{\nu}/(1-\nu)$. We have
\begin{equation*}
0<2\delta \leq \inf_{\nu \in [\delta,1]} \psi''_{\nu}(x_{\nu}),
\end{equation*}
and
\begin{equation*}
\sup_{\nu \in [\delta,1-\eta]} \psi''_{\nu}(x_{\nu}) < \infty.
\end{equation*}
Moreover, for small enough $\epsilon >0$, 
the ratio
\begin{equation*}
D(\epsilon;\delta, \eta)=\sup_{\nu\in [\delta,
1-\eta]}{\sup_{|x-x_{\nu}|<\epsilon}\left
|{\psi_{\nu}^{(3)}(x)/\psi_{\nu}''(x) }\right|}
\end{equation*}
is finite. \label{lemma:derivn}
\end{lemma}

\begin{proof}
 We can get the following first, second, third derivatives of the
 function $\psi_{\nu}(x)$:
\begin{equation*}
\psi_{\nu}'(x)=-(1-\nu)g/G+2\nu'x;
\end{equation*}
\begin{equation*}
\psi_{\nu}''(x)=-(1-\nu)(g'/G-g^2/G^2)+2\nu';
\end{equation*}
\begin{equation*}
\psi_{\nu}^{(3)}(x)=-(1-\nu)(g''/G-3g'g/G^2+2g^3/G^3);
\end{equation*}

Because $g'=(-2x)g$, $g''=(-2+4x^2)g$, and
\begin{equation*}
g(x_{\nu})/G(x_{\nu})=\frac{2\nu' x_{\nu}}{1-\nu}=z_{\nu}
\end{equation*}
at the point $x_{\nu}$, we can immediately have (\ref{eq:2nd}) and
(\ref{eq:3rd}).

Notice that $\psi_{\nu}''(x_{\nu}) \geq 2\nu'$, so it is bounded
away from zero on any interval $\nu \in [\delta, 1]$, $\delta>0$.
Also, since $x_{\nu}$ is a continuous function bounded away from
zero over $\nu$ on the interval $[\delta, 1-\eta]$ $(\delta,
\eta>0)$, we have $\psi_{\nu}''(x_{\nu})$ is also bounded above over
$[\delta, 1-\eta]$.

Now as for $\psi^{(3)}$, we note that clearly $x_{\nu}$ and $z_{\nu}$
are continuous functions on $[\delta, 1)$. And both are bounded on
the interval $\nu \in [\delta, 1-\eta]$. As a polynomial in $\nu,
x_{\nu}$ and $z_{\nu}$, $\psi_{\nu}^{(3)}$ is also bounded. If we
consider the interval $(x_{\nu}-\epsilon, x_{\nu}+\epsilon)$, the
boundness of the ratio $D(\epsilon;\delta, \eta)$ also holds
uniformly over $\nu \in [\delta, 1-\eta]$ by inspection if
$\epsilon>0$ is small enough.

\end{proof}

\section{Bounds on the Internal Angle}
\label{sec:bndinternal}

In this section, we will show how to get the internal angle decay
exponent; namely we will prove the following lemma:
\begin{lemma}
For $\epsilon>0$ and $n>n_0( \rho, \delta, \epsilon)$
\begin{equation*}
n^{-1}\log(\beta(F, G)) \leq \psi_{int}(l/n;k/l,\delta)+\epsilon,
\end{equation*}
uniformly in $l\geq \delta n$, $k\leq \rho \delta n$, $(l-k) \geq
(\nu-\delta\rho)n$.
\end{lemma}

Using the formula for the internal angle derived in the appendix, we know that
\begin{equation}
-n^{-1}\log(\beta(F, G))=-n^{-1} \log(B(\frac{1}{1+C^2k}, l-k)),
\label{eq:intexp}
\end{equation}
where
\begin{equation}
B(\alpha', m')=\theta^{\frac{m'-1}{2}} \sqrt{(m'-1)\alpha' +1}
\pi^{-m'/2} {\alpha'}^{-1/2}J(m',\theta),
\end{equation}
with $\theta=(1-\alpha')/\alpha'$ and
\begin{equation}
 J(m', \theta)=\frac{1}{\sqrt{\pi}} \int_{-\infty}^{\infty}(\int_{0}^{\infty} e^{-\theta v^2+2i v\lambda} \,dv )^{m'} e^{-\lambda^2} \,d\lambda.
\end{equation}

To evaluate (\ref{eq:intexp}), we need to evaluate the complex
integral in $J(m', \theta')$. A saddle point method based on contour
integration was sketched for similar integral expressions in
\cite{Vershik92}. A probabilistic method using large deviation
theory for evaluating similar integrals was developed in
\cite{Donoho06}. Both of these two methods can be applied in our
case and of course they will produce the same final results. In this paper we
will follow the probabilistic method from \cite{Donoho06} in this
paper. The basic idea is to see the integral in $J(m',\theta')$ as the convolution of $(m'+1)$ probability densities being
expressed in the Fourier domain. In \cite{Donoho06}, it took mechanical manipulations of the characteristic functions of the normal and half-normal distribution to arrive at this probabilistic method. In the appendix of this paper, we will give a way of deriving the internal angle formula which leads naturally to this probabilistic method and clearly explains its physical meaning.

More explicitly, we have the following lemma:
\begin{lemma}
let $\theta=(1-\alpha')/\alpha'$, where $\alpha'=\frac{1}{C^2k+1}$.
Let $T$ be a random variable with the $N(0,\frac{1}{2})$
distribution, and let $W_{m'}$ be a sum of $m'$ i.i.d. half normals
$U_{i} \sim HN(0,\frac{1}{2\theta})$. Let $T$ and $W_{m'}$ be
stochastically independent, and let $g_{T+W_{m'}}$ denote the
probability density function of the random variable $T+W_{m'}$. Then\footnote{In \cite{Donoho06}, the term $2^{-m'}$ was $2^{1-m'}$, but we believe that $2^{-m'}$ is the right term.}
\begin{equation}
B(\alpha',
m')=\sqrt{\frac{\alpha'(m'-1)+1}{1-\alpha'}}\cdot2^{-m'}\cdot
\sqrt{\pi} \cdot g_{T+W_{m'}}(0). \label{eq:pdffreq}
\end{equation}
\end{lemma}

Applying this probabilistic interpretation and large deviation
techniques, it is evaluated as in \cite{Donoho06} that
\begin{equation}
g_{T+W_{m'}} \leq \frac{2}{\sqrt{\pi}} \cdot \left(
\int_{0}^{\mu_{m'}} {v
e^{-v^2-m'\Lambda^*(\frac{\sqrt{2\theta}}{m'}v)}}\,dv+e^{-\mu_{m'}^2}
\right),
\end{equation}
where $\Lambda^*$ is the rate function for the standard half-normal
random variable $HN(0,1)$ and $\mu_{m'}$ is the expectation of
$W_{m'}$. In fact, the second term in the
sum is argued to be negligible \cite{Donoho06}. And after taking $y=\frac{\sqrt{2\theta}}{m'}v$, we have an upper bound for the first term:
\begin{equation}
\frac{2}{\sqrt{\pi}}\cdot\frac{m'^2}{2\theta} \cdot
\int_{0}^{\sqrt{2/\pi}} {y e^{-m'
(\frac{m'}{2\theta})y^2-m'\Lambda^*(y)}}\,dy. \label{eq:Grassintld}
\end{equation}

\subsection{Laplace's Method for $\psi_{int}$}
As we know, $m'$ in the exponent of (\ref{eq:Grassintld}) is defined
as $(l-k)$. Similar to evaluating the external angle
decay exponent, again we will use Laplace's method in evaluating
the internal angle decay exponent. In fact, the function
$\xi_{\gamma'}$ of (\ref{eq:intexpfun}) appears in the exponent of
(\ref{eq:Grassintld}), with $\gamma'=\frac{\theta}{m'+\theta}$.
Since $\theta=\frac{1-\alpha'}{\alpha'}=C^2k$, we have
\begin{equation*}
\gamma'=\frac{\theta}{m'+\theta}=\frac{C^2k}{(C^2-1)k+l}.
\end{equation*}
Since $k \sim \rho\delta n$, $l \sim \nu n$,
\begin{equation*}
\gamma'=\frac{k}{\frac{l}{C^2}+\frac{C^2-1}{C^2}
k}=\frac{\rho\delta}{\frac{C^2-1}{C^2}\rho\delta+ \frac{\nu}{C^2}}.
\end{equation*}

Define the function
\begin{equation*}
 f_{\gamma', m'}(y)=y e^{-m'\xi_{\gamma'}(y)},
\end{equation*}
where $\xi_{\gamma'}(y)$ is the function as defined in
(\ref{eq:intexpfun}).

If we apply similar arguments as in proving Lemma
\ref{lemma:Grassextlap}, we will
get the following lemma.
\begin{lemma}
For $\gamma' \in (0,1]$ let $y_{\gamma'} \in (0,1)$ denote the
minimizer of $\xi_{\gamma'}$. Then
\begin{equation*}
\int_{0}^{\infty} f_{\gamma', m'}(x) \,dx \leq
e^{-m'\xi_{\gamma'}(y_{\gamma'})} \cdot R_{m'}(\gamma')
\end{equation*}
where, for $\eta>0$
\begin{equation*}
m'^{-1} \sup_{\gamma \in [\eta, 1]}
\log(R_{m'}(\gamma'))=o(1)~~as~~m'\rightarrow \infty.
\end{equation*}
\end{lemma}

This means that
\begin{equation*}
g_{T+W_{m'}}(0) \leq
e^{-m'\xi_{\gamma'}(y_{\gamma'})}R_{m'}(\gamma').
\end{equation*}

So applying (\ref{eq:pdffreq}), we get
\begin{equation*}
n^{-1} \log(\beta(F,G)) \leq
\left(-\xi_{\gamma'}(y_{\gamma'})-\log(2)\right)
(\nu-\rho\delta)+o(1),
\end{equation*}
where the $o(1)$ is uniform over the range of $k$ and $l$.

\subsection{Asymptotics of $\xi_{\gamma'}$}
As in our previous discussion, we define
$\gamma'=\frac{\rho\delta}{\frac{C^2-1}{C^2}\rho\delta+
\frac{\nu}{C^2}}$, so $\gamma'$ can take any value in the range
$(0,1]$. Now we are interested in studying the asymptotics of $\xi_{\gamma'}(y_{\gamma'})$
as $\gamma' \rightarrow 0$. As in \cite{Donoho06}, using the convex duality associated
to the cumulant generating function $\Lambda(s)$ and its dual
$\Lambda^{*}$, we have
\begin{equation*}
y=\Lambda'(s),~~~~~s=({\Lambda}^*)'(y),
\end{equation*}
defining a one-one relationship $s=s(y)$ and $y=y(s)$ between $s<0$
and $0<y<\sqrt{\frac{2}{\pi}}$.

From these relations, following the same line of reasoning in
\cite{Donoho06}, we can get the minimizer $y_{\gamma'}$ of
$\xi_{\gamma'}$
\begin{equation}
\frac{1-\gamma'}{\gamma'}y_{\gamma'}=-s_{\gamma'}, \label{eq:miny}
\end{equation}
where $s_{\gamma'}=s(y_{\gamma'})$.

Because the cumulant generating function for a standard half-normal $HN(0,1)$
random variable $Y$ is $\Lambda(s)=s^2/2+\log(2\Phi(s))$, where $\phi$ and $\Phi$ are the standard density and cumulative distributions, we
have from $y=\Lambda'(s)$ that
\begin{equation}
y(s)=s\cdot(1-\frac{1}{M(s)}), s<0 \label{eq:ys}
\end{equation}
where the function of $M(s)$ is defined on $s<0$ with $0<M(s)<1$ and
$M(s) \rightarrow 1$ as $s\rightarrow -\infty$ so that
\begin{equation*}
\Phi(s)=M(s)\cdot\frac{\phi(s)}{|s|}.
\end{equation*}

Combining (\ref{eq:miny}) and (\ref{eq:ys}), we know that

\begin{equation}
M(s_{\gamma'})=1-\gamma'. \label{eq:sgamma}
\end{equation}

Further, we can derive that
\begin{equation}
\xi_{\gamma'}(y_{\gamma'})=-\frac{1}{2}{y^2_{\gamma'}}\frac{1-\gamma'}{\gamma'}-\log(2/\pi)/2+\log(y_{\gamma'}/\gamma').
\end{equation}

So by the property of the function $M(s)$ and (\ref{eq:sgamma}), as
$\gamma'\rightarrow 0$, $s_{\gamma'} \rightarrow -\infty$, we have
\begin{equation*}
 Ee^{sY}=\frac{2}{2\sqrt{\pi}} \frac{M(s)}{|s|}\sim \frac{2}{2\sqrt{\pi}}
 \frac{1}{|s|},~~~~~s\rightarrow -\infty.
\end{equation*}
By taking the logarithm for $\Lambda(s)$, $\Lambda(s)\sim -\log{|s|}$ and $\Lambda'(s)\sim -\frac{1}{s}$
as $s\rightarrow -\infty$. So by $y=\Lambda'(s)$, we have
\begin{equation*}
y(s)\sim \frac{1}{|s|}~~~s\rightarrow -\infty,
\end{equation*}
and by combining this with (\ref{eq:miny}), as $\gamma' \rightarrow 0$,
\begin{equation*}
y_{\gamma'}\sim \sqrt{\frac{\gamma'}{1-\gamma'}}.
\end{equation*}

\section{``Weak", ``Sectional" and ``Strong" Robustness}
\label{sec:wss}

So far, we have discussed the robustness of $\ell_{1}$ minimization for sparse signal recovery in the ``strong" case, namely we required robust signal recovery for all the approximately $k$-sparse signal vectors $\x$. But in applications or performance analysis, we are also often interested in the signal recovery robustness in weaker senses. As we shall see, the framework given in the previous sections can be naturally extended to the analysis of other notions of robustness for sparse signal recovery, resulting in a coherent analysis scheme.  For example, we hope to get a tighter performance bound for a particular signal vector instead of a more general, but looser, performance bound for all the possible signal vectors. In this section, we will present our null space conditions on the matrix $A$ to guarantee the performance of the programming (\ref{eq:Grassl1}) in the ``weak", ``sectional" and ``strong" senses. Here the robustness in the ``strong" sense is exactly the robustness we discussed in the previous sections.
\begin{theorem}
Let $A$ be a general $m\times n$ measurement matrix, $\x$ be an $n$-element
vector and $\y=A\x$. Denote $K$ as a subset of $\{1,2,\dots,n\}$ such
that its cardinality $|K|=k$ and further denote $\overline{K}=\{1,2,\dots,n\}\setminus K$. Let $\w$
denote an $n \times 1$ vector. Let $C>1$ be a fixed number.

\begin{itemize}
\item  (Weak Robustness) Given a specific set $K$ and suppose that the part of $\x$ on $K$, namely $\x_{K}$ is fixed.
$\forall \x_{\overline{K}}$, any solution $\hat{\x}$ produced by (\ref{eq:Grassl1})
satisfies
\begin{equation*}
 \|\x_K\|_1-\|\hat{\x}_K\|_{1} \leq
\frac{2}{C-1} \|\x_{\overline{K}}\|_1
\end{equation*}
 and
\begin{equation*}
 \|(\x-\hat{\x})_{\overline{K}}\|_1
\leq \frac{2C}{C-1} \|\x_{\overline{K}}\|_1,
\end{equation*}
if and only if $\forall \w\in \mathbb{R}^n~\mbox{such that}~A\w=0$,
we have
 \begin{equation}
\|\x_K+\w_{K}\|_1+ \|\frac{\w_{\overline{K}}}{C}\|_1 \geq \|\x_K\|_1;
 \label{eq:Grasswthmeq1}
 \end{equation}
\item (Sectional Robustness)
Given a specific set $K \subseteq \{1,2,\dots,n\}$. Then $\forall \x \in \mathbb{R}^n$, any solution $\hat{\x}$ produced by
(\ref{eq:Grassl1}) will satisfy
\begin{equation*}
 \|\x-\hat{\x}\|_{1} \leq \frac{2(C+1)}{C-1}  \| \x_{\overline{K}}
 \|_{1},
\end{equation*}
if and only if $\forall \x' \in \mathbb{R}^n$, $\forall \w\in
\mathbb{R}^n~\mbox{such that}~A\w=0$,
 \begin{equation}
 \|\x'_K+\w_{K}\|_1+ \|\frac{\w_{\overline{K}}}{C}\|_1 \geq \|\x'_K\|_1;
 \label{eq:wthmeq2}
 \end{equation}
\item(Strong Robustness)
If for all possible $K \subseteq \{1,2,\dots,n\}$, and for all $\x \in
\mathbb{R}^n$, any solution $\hat{\x}$ produced by
(\ref{eq:Grassl1}) satisfies
\begin{equation*}
 \|\x-\hat{\x}\|_{1} \leq \frac{2(C+1)}{C-1}  \| \x_{\overline{K}}
 \|_{1},
\end{equation*}
 if and only if $\forall K \subseteq \{1,2,\dots,n\}, \forall \x' \in \mathbb{R}^n$, $\forall \w\in \mathbb{R}^n~\mbox{such that}~A\w=0$,
 \begin{equation}
\|\x'_K+\w_{K}\|_1+ \|\frac{\w_{\overline{K}}}{C}\|_1 \geq \|\x'_K\|_1.
 \label{eq:wthmeq3}
 \end{equation}
\end{itemize}
\end{theorem}

\begin{proof} We will first show the sufficiency of the null space conditions for the various definitions of robustness.
Let us begin with the ``weak" robustness part. Let $\w=\hat{\x}-\x$ and we must have $A\w=A(\hat{\x}-\x)=0$. From the
triangular inequality for $\ell_{1}$ norm and the fact that $\|\x\|_{1} \geq \| \x+\w \|_{1}$, we have
\begin{eqnarray*}
&&~~\|\x_K\|_1-\|\x_K+\w_K\|_1 \\
&&\geq \|\w_{\overline{K}}+\x_{\overline{K}}\|_1-\|\x_{\overline{K}}\|_1 \\
&&\geq \|\w_{\overline{K}}\|_1-2\|\x_{\overline{K}}\|_1.
\end{eqnarray*}
But the condition (\ref{eq:Grasswthmeq1}) guarantees that
\begin{equation*}
\|\w_{\overline{K}}\|_1\geq C(\|\x_K\|_1-\|\x_K+\w_K\|_1),
\end{equation*}
so we have
\begin{equation*}
 \|\w_{\overline{K}}\|_1
\leq \frac{2C}{C-1} \|\x_{\overline{K}}\|_1,
\end{equation*}
and
\begin{equation*}
\|\x_K\|_1-\|\hat{\x}_K\|_{1} \leq \frac{2}{C-1} \|\x_{\overline{K}}\|_1.
\end{equation*}
For the ``sectional" robustness, again, we let $\w=\hat{\x}-\x$.
Then there must exist an $\x'\in \mathbb{R}^n$ such that
\begin{equation*}
\|\x'_K+\w_{K}\|_1=\|\x'_K\|_1-\|\w_{K}\|_1.
\end{equation*}
Following the condition (\ref{eq:wthmeq2}), we have
\begin{equation*}
\|\w_{K}\|_1 \leq \|\frac{\w_{\overline{K}}}{C}\|_1.
\end{equation*}
Since
\begin{equation*}
\|\x\|_{1} \geq \| \x+\w \|_{1},
\end{equation*}
following the proof of Theorem $1$, we have
\begin{equation*}
\|\x-\hat{\x}\|_{1} \leq \frac{2(C+1)}{C-1}  \| \x_{\overline{K}}\|_1.
\end{equation*}
The sufficiency of the condition (\ref{eq:wthmeq3}) for strong
robustness also follows.

Necessity: Since in the proof of the sufficiency, equalities can be
achieved in the triangular equalities, the conditions
(\ref{eq:Grasswthmeq1}), (\ref{eq:wthmeq2}) and (\ref{eq:wthmeq3})
are also necessary conditions for the the respective robustness to
hold for every $\x$ (otherwise, for certain $\x$'s, there will be
$\x'=\x+\w$ with $\|\x'\|_1<\|\x\|_1$ while violating the respective
robustness definitions. Also, such $\x'$ can be the solution to
($\ref{eq:Grassl1}$)). The detailed arguments will similarly follow the proof of the second part of
Theorem \ref{thm:th1}.

\end{proof}

The conditions for ``weak", ``sectional" and ``strong" robustness
seem to be very similar, and yet there are indeed huge differences. The ``weak" robustness condition is for
$\x$ with a specific $\x_{K}$ on a specific subset $K$, the
``sectional" robustness condition is for $\x$ with all possible
$\x_K$'s on a specific subset $K$, and the ``strong" robustness
conditions are for $\x$'s with all possible $\x_K$'s on all possible
subsets. Basically, the ``weak'' robustness condition
(\ref{eq:Grasswthmeq1}) guarantees that the $\ell_{1}$ norm of
$\hat{\x}_K$ is not too far away from the $\ell_{1}$ norm of $\x_K$
and the error vector $\w_{\overline{K}}$ is small in $\ell_{1}$ norm when
$\|\x_{\overline{K}}\|_1$ is small. Notice that if we define
\begin{equation*}
\kappa=\max_{A\w=0, \w \neq 0}
\frac{\|\w_K\|_1}{\|\w_{\overline{K}}\|_1},
\end{equation*}
then
\begin{equation*}
\|\x-\hat{\x}\|_1 \leq \frac{2C(1+\kappa)}{C-1} \|\x_{\overline{K}}\|_1.
\end{equation*}
That means, if $\kappa$ is not $\infty$  for a measurement matrix
$A$, $\|\x-\hat{\x}\|_1 $ is also small when $\|\x_{\overline{K}}\|_1$ is
small. Indeed, it is not hard to see that, for a given matrix $A$, $\kappa < \infty$
as long as the rank of matrix $A_{K}$ is equal to
$|K|=k$, which is generally satisfied for $k < m$.

While the ``weak" robustness condition is only for
one specific signal $\x$,  the ``sectional" robustness condition
instead guarantees that given \emph{any} approximately $k$-sparse
signal mainly supported on the subset $K$, the
$\ell_{1}$-minimization gives a solution $\hat{\x}$ close to the
original signal by satisfying (\ref{eq:noiseadj}). When we measure
an approximately $k$-sparse signal $\x$ (the support of the $k$
largest-magnitude components is fixed though unknown to the decoder)
using a randomly generated measurement matrix $A$, the ``sectional"
robustness conditions characterize the probability that the
$\ell_{1}$ minimization solution satisfies (\ref{eq:noiseadj}) for
\emph{any} signals for the set $K$. If that probability goes to $1$
as $n\rightarrow \infty$ for any subset $K$, we know that there
exist measurement matrices $A$'s that guarantee (\ref{eq:noiseadj})
on ``almost all" support sets (namely, (\ref{eq:noiseadj}) is
``almost always" satisfied). The ``strong" robustness condition
instead guarantees the recovery for approximately sparse
signals mainly supported on \emph{any} subset $K$. The ``strong"
robustness condition is useful in guaranteeing the decoding bound
\emph{simultaneously} for \emph{all} approximately $k$-sparse signals
under a single measurement matrix $A$.

Interestingly, after we take $C=1$ and let (\ref{eq:Grasswthmeq1}),
(\ref{eq:wthmeq2}) and (\ref{eq:wthmeq3}) take strict inequality for
all $\w\neq 0$ in the null space of $A$, the conditions
(\ref{eq:Grasswthmeq1}), (\ref{eq:wthmeq2}) and (\ref{eq:wthmeq3})
are also sufficient and necessary conditions for unique exact
recovery of ideally $k$-sparse signals in ``weak", ``sectional" and
``strong" senses \cite{Donoho06}, namely the unique exact recovery
of a specific ideally $k$-sparse signal, the unique exact recoveries
of all ideally $k$-sparse signal on a specific support set $K$ and
the unique exact recoveries of all ideally $k$-sparse signal on all
possible support sets $K$. In fact, if $\|\x_{\overline{K}}\|_1=0$, from
similar triangular inequality derivations in Theorem $1$, we have
$\hat{\x}=\x$ under all the three conditions.

For a given value $\delta=\frac{m}{n}$ and any value $C\geq1$, we
will determine the value of feasible $\zeta=\frac{k}{n}$ for which
there exist a sequence of $A$'s such that these three conditions
are satisfied when $n \rightarrow \infty$ and $\frac{m}{n}=\delta$.
As manifested by the statements of the three conditions
(\ref{eq:Grasswthmeq1}), (\ref{eq:wthmeq2}) and (\ref{eq:wthmeq3})
and the previous discussions in Section \ref{sec:Grassprobnull}, we
can naturally extend the Grassmann angle approach to analyze the
bounds for the probabilities that (\ref{eq:Grasswthmeq1}),
(\ref{eq:wthmeq2}) and (\ref{eq:wthmeq3}) fail. Here we will denote
these probabilities as $P_1$, $P_2$ and $P_3$ respectively. Note
that there are $\binom{n}{k}$ possible support sets $K$ and there
are $2^k$ possible sign patterns for signal $\x_K$. From previous
discussions, we know that the event that the condition
(\ref{eq:Grasswthmeq1}) fails is the same for all $\x_K$'s of a
specific support set and a specific sign pattern. Then following the
same line of reasoning as in Section \ref{sec:Grassprobnull}, we
have
 \begin{eqnarray}
  &&P_1=P_{K,-}\\
  &&P_2\leq 2^k \times P_1,\\
  &&P_3 \leq \binom{n}{k} \times 2^k \times P_1,
  \label{eq:wnum23}
\end{eqnarray}
where $P_{K,-}$ is the probability as in (\ref{eq:union}).

We have the following lemma about the $P_1$, $P_2$ and $P_3$:
\begin{lemma}
For any $C>1$, we define $\zeta_{W}(\delta)$, $\zeta_{Sec}(\delta)$,
and $\zeta_{S}(\delta)$ to be the largest fraction
$\zeta=\frac{k}{n}$ such that the condition (\ref{eq:Grasswthmeq1})
(\ref{eq:wthmeq2}) and (\ref{eq:wthmeq3}) are satisfied with
overwhelming probability as $n \rightarrow 0$ if we sample the
$(n-m)$-dimensional null space uniformly, where
$\frac{m}{n}=\delta$. Then
\begin{eqnarray*}
\zeta_{W}(\delta)&>&0,\\
\zeta_{Sec}(\delta)&>&0,\\
\zeta_{S}(\delta)&>&0
\end{eqnarray*}
 for any $C>1$ and $\delta>0$. Also,
\begin{equation*}
\lim_{\delta \rightarrow 1} \zeta_{W}(\delta)=1
\end{equation*}
for any $C>1$. \label{lemma:wthlimit}
\end{lemma}

The proof of this lemma is listed in the appendix.  It is worthwhile mentioning that the formula for $P_1$ is exact since there is no union bound involved and so the threshold bound for the ``weak" robustness is tight. In a short summary, the results in this section suggest that even if $k$ is
very close to the weak threshold for ideally sparse signals, we can still have robustness results for approximately sparse signals while
the results using restricted isometry conditions \cite{Candes05b} may suggest smaller sparsity level for recovery robustness. This is the first such a kind of result. The numerical results of $\zeta$ making sure that $P_1$, $P_2$, $P_3$ converge to zero overwhelmingly are presented in Section \ref{sec:numericalold}.

\section{Analysis of $\ell_{1}$ Minimization under Noisy Measurements}
\label{sec:noisymeas} In the previous sections, we have analyzed the
$\ell_{1}$ minimization algorithm for decoding general signals. In this section, we will discuss the effect of noisy
measurements on the $\ell_{1}$ minimization of general signals, using the null space characterization.

\begin{theorem}
Assume that $A$ is a general $m \times n$ measurement matrix $A$ with rank $m$ and its minimum nonzero singular value is denoted as
$\sigma_{\text{min}}$. Further, assume that $\y=A\x+\b$, with its
$\ell_{2}$-norm $\|\b\| \leq \epsilon$,   and that $\w$ is an
$n\times 1$ vector. Let $K$ be any subset of $\{1,2,\dots,n\}$ such
that its cardinality $|K|=k$ and let $K_i$ denote the $i$-th element
of $K$. Further, let $\overline{K}=\{1,2,\dots,n\} \setminus K$. Then the
solution $\hat{\x}$ produced by (\ref{eq:Grassl1}) will satisfy
\begin{equation*}
\|\x-\hat{\x}\|_{1} \leq  \frac{2(C+1)}{C-1}  \| \x_{\overline{K}}
\|_{1}+\frac{(3C+1)\sqrt{n} \epsilon}{(C-1)\sigma_{\text{min}}}
\end{equation*}
with $C>1$, if $\forall \w \in \mathbb{R}^n$ such that
\begin{equation*}
~A\w=0~
\end{equation*}
and for all the subsets $K$ with $|K|=k$, we have
\begin{equation}
 C \sum_{i=1}^{k}|\w_{K_i}|\leq\sum_{i=1}^{n-k}|\w_{\overline{K}_i}|.
\label{eq:noisythmeq}
\end{equation}
\label{thm:noisym}
\end{theorem}

\begin{proof}
Since
\begin{equation*}
  \y=A\x+\b,
\end{equation*}
we can write
\begin{equation*}
  \y=A\x^{*},
\end{equation*}
where
\begin{equation*}
  \|\x^{*}-\x\| \leq \frac{\epsilon}{\sigma_{\text{min}}}.
\end{equation*}

By the Cauchy-Schwarz inequality, we have
\begin{equation*}
  \|\x^{*}-\x\|_1 \leq \frac{\sqrt{n} \epsilon}{\sigma_{\text{min}}}.
\end{equation*}

Suppose the matrix $A$ has the claimed null space property. Now the
solution $\hat{\x}$ of (\ref{eq:Grassl1}) satisfies
$\|\hat{\x}\|_{1} \leq \|\x^{*}\|_{1}$. Since $A \hat{\x}=\y$, it
easily follows that $\w=\hat{\x}-\x^{*}$ is in the null space of
$A$. Therefore we can further write $\|\x^{*}\|_{1} \geq \|
\x^{*}+\w \|_{1}$. Using  the triangular inequality for the
$\ell_{1}$ norm we obtain
\begin{eqnarray*}
\|\x_{K}^{*}\|_{1}+\|\x_{\overline{K}}^{*}\|_{1}&=& \|\x^{*}\|_{1}\nonumber\\
 &\geq& \| \hat{\x} \|_{1} =\| \x^{*}+\w \|_{1} \nonumber \\
 & \geq &\|\x_{K}^{*}\|_{1}-\|\w_{K}\|_{1}+\|\w_{\overline{K}}\|_{1}-\|\x_{\overline{K}}^{*}\|_{1}\nonumber \\
 & \geq&\|\x_{K}^{*}\|_{1}-\|\x_{\overline{K}}^{*}\|_{1}+\frac{C-1}{C+1}\|\w\|_{1},\nonumber
\end{eqnarray*}
where the last inequality is from the claimed null space
property. Relating the first equality and the last inequality above,
we have $2\|\x_{\overline{K}}^{*}\|_{1} \geq
\frac{(C-1)}{C+1}\|\w\|_{1}$.

Since
\begin{equation*}
\|\x_{\overline{K}}^{*}\|_{1} \leq \|\x_{\overline{K}}\|_{1}+
\|\x^{*}-\x\|_{1},
\end{equation*}
we get
\begin{eqnarray*}
\|\w\|_{1} &\leq& \frac{2(C+1)}{C-1}  \| \x_{\overline{K}}^{*} \|_{1}\\
&\leq& \frac{2(C+1)}{C-1}  \| \x_{\overline{K}} \|_{1}+\frac{2(C+1)}{C-1}
\| \x^{*}-\x\|_{1}.
\end{eqnarray*}

From the triangular inequality,
\begin{eqnarray}
\|\x-\hat{\x}\|_{1} &\leq& \|\x-\x^{*}\|_1+\|\w\|_{1}\\
&\leq& \frac{2(C+1)}{C-1}  \| \x_{\overline{K}} \|_{1}+\frac{3C+1}{C-1}
\|
\x^{*}-\x\|_{1},\\
&\leq& \frac{2(C+1)}{C-1}  \| \x_{\overline{K}}
\|_{1}+\frac{(3C+1)\sqrt{n} \epsilon}{(C-1)\sigma_{\text{min}}}.
\end{eqnarray}
\end{proof}

If the elements in the measurement matrix $A$ are i.i.d. as the unit
real Gaussian random variables $N(0,1)$, following upon the work of
Marchenko and Pastur \cite{Marcenko67}, Geman\cite{Geman80} and
Silverstein \cite{silver} proved that for $m/n=\delta$, as $n
\rightarrow \infty$,
\begin{equation*}
\frac{1}{\sqrt{n}}\sigma_{min} \rightarrow 1-\sqrt{\delta}
\end{equation*}
almost surely as $n \rightarrow \infty$.

Then almost surely as $n \rightarrow \infty$, $\frac{(3C+1)\sqrt{n}\epsilon}{(C-1)\sigma_{\text{min}}} \rightarrow \frac{(3C+1)
\epsilon}{(C-1)(1-\sqrt{\delta})}$. So in this case, we have $\|\x^{*}-\x\|_{1}$ is upper-bounded by some constant times
$\epsilon$. It is also worth mentioning that the error bound derived above is for a plain $\ell_1$ minimization optimization programming, which does not use any prior knowledge of the magnitudes of the noise in the computations, while the error bounds in the literatures often assume that such information is known and is used in the convex programming algorithms. To get an error bound in terms of $\ell_2$ norm, we can invoke the almost Euclidean property of the null space,namely every vector $\w$ has an $\ell_2$ norm scaling as $O(\frac{1}{\sqrt{n}})$ of its $\ell_1$ norm. Though we choose not to do it in detail in this paper, it is easy to see that the error bound here has the same scaling in $\epsilon$ as the analysis through the restricted isometry property \cite{CRT06}; however, this analysis is warranted even when the cardinality $|K|$ of the set $K$ is much larger than the known cardinality bounds for the restricted isometry property. It is also possible to extend the concepts of ``weak", ``sectional" and ``strong" robustness analysis to noisy measurements, which will also similarly show that even if the cardinality of the set $K$ is very close to the ``weak" threshold for the ideally sparse signals, we can still have the robustness of $\ell_1$ minimization to noisy measurements.

\section{Numerical Computations on the Bounds of $\zeta$}
\label{sec:numericalold}

In this section, we will numerically evaluate the performance bounds
on $\zeta=\frac{k}{n}$ such that the conditions (\ref{eq:thmeq}),
(\ref{eq:Grasswthmeq1}), (\ref{eq:wthmeq2}) and (\ref{eq:wthmeq3})
are satisfied with overwhelming probability as $n \rightarrow
\infty$.

First, we know that the condition (\ref{eq:thmeq}) fails with probability
 \begin{equation}
  P\leq \binom{n}{k} \times 2^k \times 2\times \sum_{s \geq 0}\sum_{G \in \Im_{m+1+2s}(\text{SP})}
{\beta(F,G)\gamma(G,\text{SP})}, \label{eq:num}
\end{equation}

Recall that we assume $\frac{m}{n}=\delta$, $l=(m+1+2s)+1$ and
$\nu=\frac{l}{n}$.  In order to make $P$ overwhelmingly converge to
zero as $n\rightarrow \infty$, following the discussions in Section
\ref{sec:evathebound}, one sufficient condition is to make sure that
the exponent for the combinatorial factors
\begin{equation}
\psi_{com}= \lim_{n\rightarrow \infty}{\frac{ \log{(\binom{n}{k} 2^k
2\binom{n-k}{l-k} 2^{l-k})}}{n} }
\end{equation}

and the negative exponent for the angle factors

\begin{equation}
\psi_{angle}= -\lim_{n\rightarrow \infty}{
\frac{\log{(\beta(F,G)\gamma(G,\text{SP}))}}{n} }
\end{equation}

satisfy $\psi_{com}-\psi_{angle}<0$ uniformly over $\nu \in [\delta,
1)$.

Following \cite{Donoho06} we take $m=0.5555n$. By analyzing the
decaying exponents of the external angles and internal angles
through the Laplace methods as in Section \ref{sec:bndexnangle} and
\ref{sec:bndinternal}, we can compute the numerical results as shown
in Figure \ref{fig:sl3}, Figure \ref{fig:wss} and Figure
\ref{fig:wss01}. In Figure \ref{fig:sl3}, we show the largest
sparsity level $\zeta=\frac{k}{n}$ (as a function of $C$) which
makes the failure probability of the condition ($\ref{eq:Grassxcondition}$)
approach zero asymptotically as $n\rightarrow \infty$. As we
can see, when $C=1$, we get the same bound  $\zeta=0.095 \times
0.5555 \approx 0.0528$ as obtained for the ``weak" threshold for the ideally sparse signals in
\cite{Donoho06}. As expected, as $C$ grows, the $\ell_{1}$
minimization requires a smaller sparsity level $\zeta$ to achieve
higher signal recovery accuracy.


\begin{figure}[htb]
\centering \centerline{\epsfig{figure=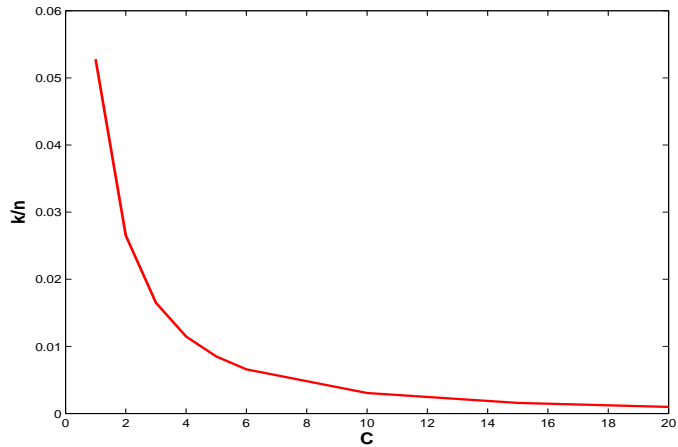,
width=9cm,height=6cm}}
\caption{Allowable sparsity as a function of $C$ (allowable
imperfection of the recovered signal is $\frac{2(C+1)\Delta}{C-1}$)}
\label{fig:sl3}
\end{figure}

\begin{figure}[htb]
\psfrag{n}{$\nu$} \centering
\centerline{\epsfig{figure=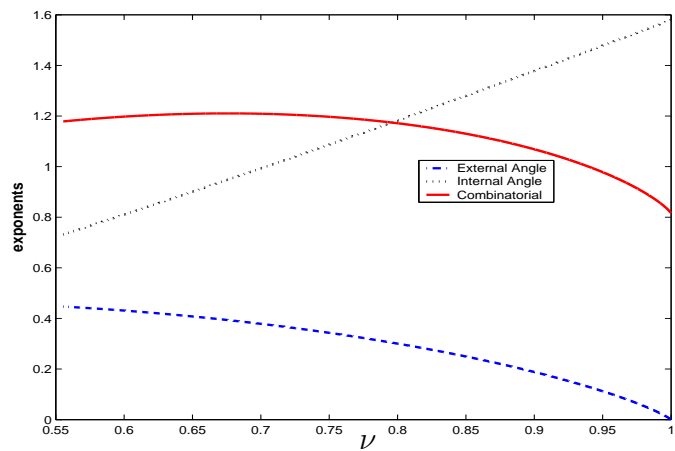,width=9cm,height=6cm}}
\caption{The Combinatorial, Internal and External Angle Exponents}
 \label{fig:Grassexponents}
\end{figure}

\begin{figure}[htb]
\psfrag{n}{$\nu$} \centering
\centerline{\epsfig{figure=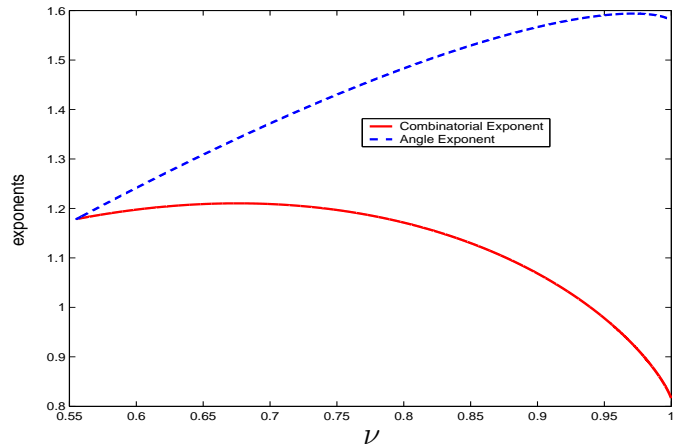,width=9cm,height=6cm}}
\caption{The Combinatorial Exponents and the Angle Exponents}
\label{fig:exponents2}
\end{figure}

\begin{figure}[htb]
\centering
\centerline{\epsfig{figure=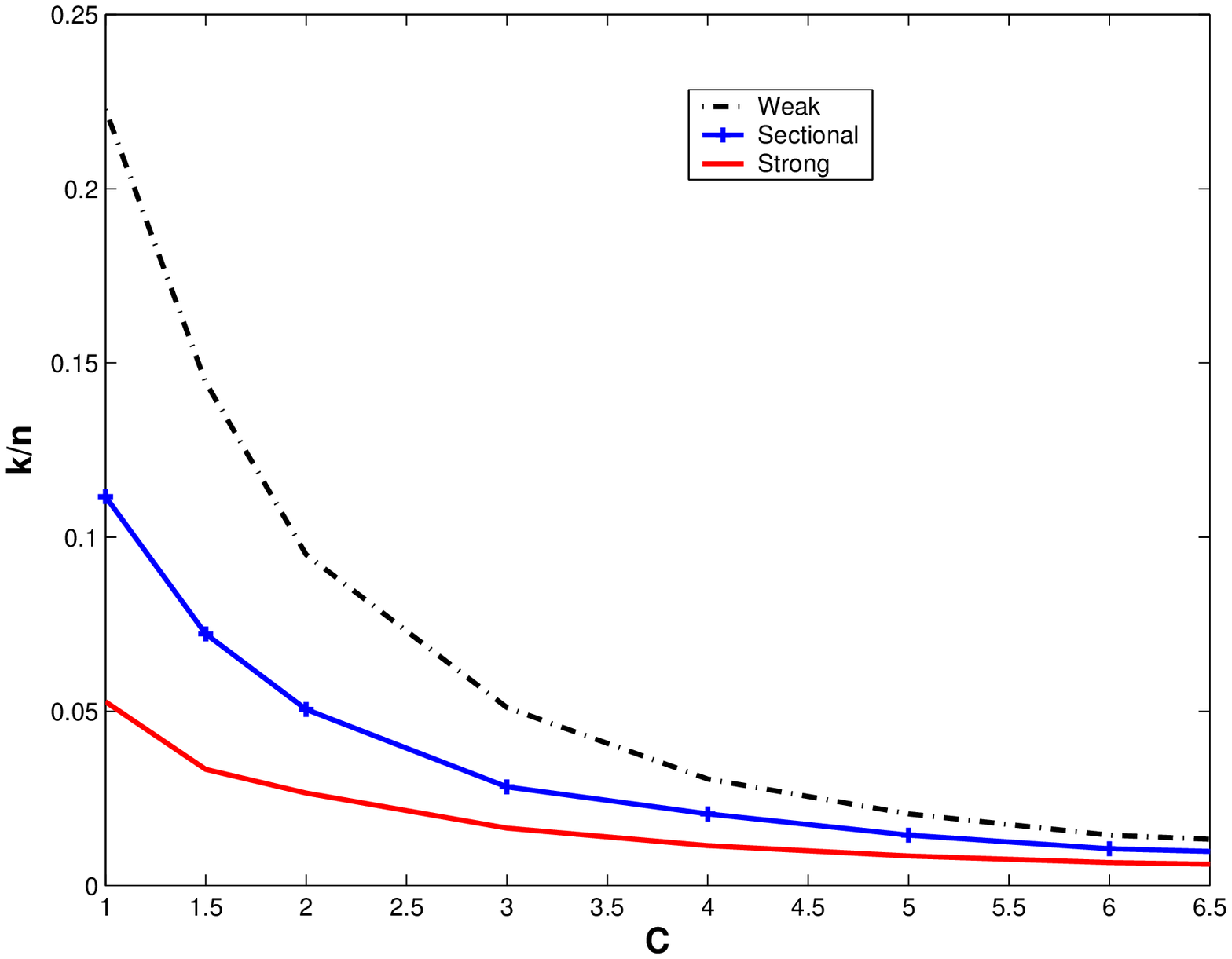,width=9cm,height=6cm}}
\caption{The Weak, Sectional and Strong Robustness Bounds}
\label{fig:wss}
\end{figure}

\begin{figure}[htb]
\psfrag{n}{$\delta$} \psfrag{r}{$\rho$} \centering
\centerline{\epsfig{figure=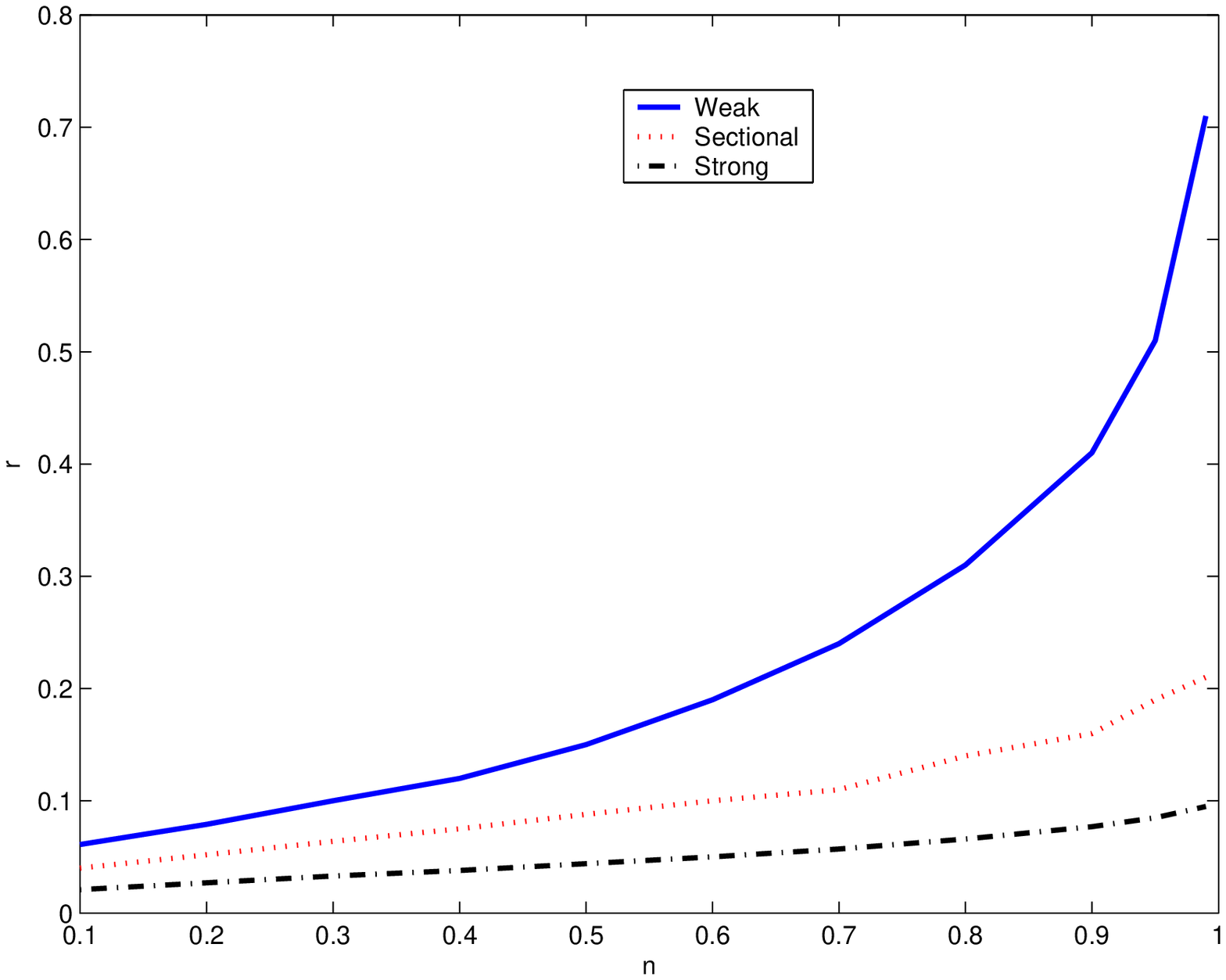,width=9cm,height=6cm}}
\caption{The Weak, Sectional and Strong Robustness Bounds}
\label{fig:wss01}
\end{figure}

In Figure \ref{fig:Grassexponents}, we show the exponents $\psi_{com}$,
$\psi_{int}$, $\psi_{ext}$ under the parameters $C=2$,
$\delta=0.5555$ and $\zeta=0.0265$. For the same set of parameters,
in Figure \ref{fig:exponents2}, we compare the exponents
$\psi_{com}$ and $\psi_{angle}$: the solid curve denotes
$\psi_{angle}$ and the dashed curve denotes $\psi_{com}$. It shows
that, under $ \zeta=0.0265$, $\psi_{com}-\psi_{angle}<0$ uniformly
over $\delta \leq \nu \leq 1$. Indeed, $\zeta=0.0265$ is the bound
shown in Figure \ref{fig:sl3} for $C=2$. In Figure \ref{fig:wss},
for the parameter $\delta=0.5555$, we give the bounds $\zeta$ as a
function of $C$ for satisfying the signal recovery robustness
conditions (\ref{eq:Grasswthmeq1}), (\ref{eq:wthmeq2}) and
(\ref{eq:wthmeq3}) respectively in the ``weak", ``sectional" and
``strong" senses. In Figure \ref{fig:wss01}, fixing $C=2$, we plot
how large $\rho=\zeta/\delta$ can be for different $\delta$'s while
satisfying the signal recovery robustness conditions
(\ref{eq:Grasswthmeq1}), (\ref{eq:wthmeq2}) and (\ref{eq:wthmeq3})
respectively in ``weak", ``sectional" and ``strong" senses.

\section{Conclusion}
\label{sec:con}

It is well known that $\ell_{1}$ optimization can be used to recover
ideally sparse signals in compressive sensing, if the underlying signal is sparse enough. While for the
ideally sparse signals, the results of \cite{Donoho06} have given us very sharp bounds on the sparsity threshold the $\ell_1$ minimization can recover, sharp bounds for the recovery of general signals or approximately sparse signals were not available.

In this paper we analyzed a null space characterization of the measurement matrices for the performance bounding of $\ell_{1}$-norm optimization for general signals or approximately sparse. Using high-dimensional geometry tools, we
give a unified \emph{null space Grassmann angle}-based analytical
framework for compressive sensing. This new framework gives sharp
quantitative tradeoffs between the signal sparsity parameter and the recovery
accuracy of the $\ell_{1}$ optimization for general signals or approximately sparse
signals. As expected, the neighborly polytopes result of \cite{Donoho06} for ideally sparse
signals can be viewed as a special case on this tradeoff curve.  It can therefore be of practical use in applications where
the underlying signal is not ideally sparse and where we are
interested in the quality of the recovered signal. For example, using the results and their extensions in this paper and \cite{Donoho06}, we are able to give a precise sparsity threshold analysis for weighted $\ell_1$ minimization when prior information about the signal vector is available \cite{isitweighted}.  In \cite{icasspiterweighted}, using the robustness result from this paper, we are able to show that a polynomial-time iterative weighted $\ell_1$ minimization algorithm can provably improve over the sparsity threshold of $\ell_1$ minimization for interesting classes of signals, even when prior information is not available.

In essence, this work investigates the fundamental ``balancedness" property of linear subspaces, and may be of independent mathematical interest. In future work, it is interesting to obtain more accurate analysis for compressive sensing under noisy measurements than presented in the current paper.

\section{Appendix}
\subsection{Some Concepts in the High Dimensional Geometry}
In this part, we will give the explanations of several often used geometric terminologies in this paper for the purpose of quick reference.
\subsubsection{the Grassmann Manifold}
The Grassmann manifold $\text{Gr}_{i}(j)$ refers to the set of $i$-dimensional subspaces in the $j$-dimensional Euclidean space $\mathbb{R}^j$. It is known that there exists a unique invariant measure $\mu'$ on $\text{Gr}_{i}(j)$ such that $\mu' (\text{Gr}_{i}(j))$=1.

For more facts on the Grassmann manifold, please see \cite{Grass}.
\subsubsection{Polytope, Face, Vertex} A polytope in this paper refers to the convex hull of a finite number points in the Euclidean space. Any extreme point of a polytope is a vertex of this polytope. A face of a polytope is defined as the convex hull of a set of its vertices such that no point in this convex hull is an interior point of the polytope. The dimension of a face refers to the dimension of the affine hull of that face. The book \cite{Grunbaumbook} offers a nice reference on the convex polytopes.
\subsubsection{Cross-polytope}
The $n$-dimensional cross-polytope is the polytope of unit $\ell_1$ ball, namely it is the set
\begin{equation*}
\{\x\in \mathbb{R}^n ~|~\|\x\|_1=1 \}.
\end{equation*}
The $n$-dimensional cross-polytope has $2n$ vertices, namely $\pm e_1, \pm e_2, ..., \pm e_n$, where $e_i$, $1\leq i \leq n$, is the unit vector with its $i$-th coordinate element being 1. Any $k$ extreme points without opposite pairs at the same coordinate will constitute a $(k-1)$-dimensional face of the cross-polytope. So the cross-polytope will have $2^k\binom{n}{k}$ faces of dimension $(k-1)$.

\subsubsection{the Grassmann Angle}
The Grassmann angle for a $n$-dimensional cone $\mathfrak{C}$ under the Grassmann manifold $\text{Gr}_{i}(n)$, is the measure of the set of $i$-dimensional subspaces (over $\text{Gr}_{i}(n)$) which intersect the cone $\mathfrak{C}$ nontrivially (namely at some other point besides the origin). For more details on the Grassmann angle, internal angle, and external angle, please refer to \cite{Grunbaum68}\cite{Grunbaumbook}\cite{McMullen1975}.

\subsubsection{the Internal Angle}
An internal angle $\beta(F_1, F_2)$, between two faces $F_1$ and $F_2$ of a polytope or a polyhedral cone, is the fraction of the
hypersphere $S$ covered by the cone obtained by observing the face
$F_2$ from the face $F_1$.The internal angle $\beta(F_1, F_2)$ is
defined to be zero when $F_1 \nsubseteq F_2$ and is defined to be
one if $F_1=F_2$.  Note the dimension of the
hypersphere $S$ here matches the dimension of the corresponding cone
discussed. Also, the center of the hypersphere is the apex of the
corresponding cone. All these defaults also apply to the definition
of the external angles.
\subsubsection{the External Angle}
An external angle $\gamma(F_3, F_4)$, between two faces $F_3$ and $F_4$ of a polytope or a polyhedral cone, is the fraction of the
hypersphere $S$ covered by the cone of outward normals to the
hyperplanes supporting the face $F_4$ at the face $F_3$. The
external angle $\gamma(F_3, F_4)$ is defined to be zero when $F_3
\nsubseteq F_4$ and is defined to be one if $F_3=F_4$.

\subsection{Proof of Lemma \ref{lemma:nullspaceGaussian}}
\begin{proof}
The first statement is obvious since multiplying $A$ with a unitary $\Theta$ keeps the columns independent and the entries i.i.d. Gaussian.

Now let us look at the proof of the second statement. Consider the Singular Value Decomposition (SVD) $A = U \Sigma V^*$,
where $U$ and $V$ have orthonormal columns and $\Sigma$ is diagonal. Consider now $A\Theta$, for any given deterministic unitary $\Theta$: $A\Theta = U \Sigma V^*\Theta$. This is clearly the SVD of $A\Theta$; in particular, $\Theta^*V$ represents the right singular vectors of $A\Theta$. Since $A$ and $A\Theta$ have the same distribution (for all unitary $\Theta$), the same must be true of the right singular vectors $V$ and $\Theta^*V$. Therefore the distribution of $V$ is left-rotationally invariant: $P_V(V) =P_V(\Theta^*V)$. Now the null space of $A$ can be written as $Z = V^\perp X$, where $V^\perp$ is an $n \times (n-m)$ matrix with orthonormal columns that are orthogonal to $V$, i.e., $V^* V^\perp = 0$, and $X$ is any invertible $(n-m) \times (n-m)$ matrix. Now it is easy to see that if we change $V$ to $\Theta^*V$, for any unitary $\Theta$, we must change $V^\perp$ to $\Theta^*V^\perp$. But since left-multiplication by a unitary $\Theta^*$ does not change the distribution of $V$, left multiplication by a unitary $\Theta^*$ must not change the distribution of $V^\perp$. Thus $V^\perp$, and by fiat $Z = V^\perp X$, are left-rotationally invariant. Note to simplify the arguments, we have so far assumed that the matrix $A$ is of full rank $m$, which is true with probability 1. However, we should note that these arguments also work when the matrix $A$ is rank-deficient.

Let $G$ be an $n \times (n-m)$ matrix with i.i.d. $\mathcal{N}(0,1)$ entries and consider the QR decomposition:
$G = QR$,
where $Q$ is an $n \times (n-m)$ matrix with orthonormal columns and $R$ is an $(n-m) \times (n-m)$ upper triangular matrix with non-negative diagonals. Then it is well known that $Q$ has a left-rotationally invariant distribution, and that $R$ is a random matrix, independent of $Q$, whose strictly upper triangular entries are i.i.d. $\mathcal{N}(0,1)$ and whose $i$-th diagonal entry is an independent Chi-square random variable with $(n-i+1)/2$ degrees of freedom \cite{Muirhead}.
This implies that we can always take the $V^\perp$ obtained from the 2nd statement and post-multiply it by an independent upper triangular $R$ (with the aforementioned distribution) to obtain a matrix
$Z = V^\perp R$ with i.i.d. $\mathcal{N}(0,1)$ entries. It is always possible to choose a basis $Z$ for the null space such that $Z$ has i.i.d. ${\cal N}(0,1)$ entries.

\end{proof}

\subsection{Derivation of the Internal Angles}
There are two situations in the derivations of the internal angles $\beta(F,G)$ for the skewed cross-polytope: when $G$ is a regular face and when $G$ is the whole skewed cross-polytope $\text{SP}$. These two cases are respectively dealt with in Lemma \ref{lemma:internalapp} and Lemma \ref{lemma:internalSP}.

\begin{lemma}
Suppose that $F$ is a $(k-1)$-dimensional face of the skewed
cross-polytope
\begin{equation*}
\text{SP}=\{\y\in \mathbb{R}^n~|~\|\y_K\|_1+ \|\frac{\y_{\overline{K}}}{C}\|_1
\leq 1\}
\end{equation*}
supported on the subset $K$ with $|K|=k$. Then the internal angle
$\beta(F,G)$ between the $(k-1)$-dimensional face $F$ and a
$(l-1)$-dimensional face $G$ ($F \subseteq G$, $G \neq $SP$ $) is
given by
\begin{equation}
\beta(F,G)=\frac{V_{l-k-1}(\frac{1}{1+C^2k},l-k-1)}{V_{l-k-1}(S^{l-k-1})},
\label{eq:interfrac}
\end{equation}
where $V_i(S^i)$ denotes the $i$-th dimensional surface measure on
the unit sphere $S^{i}$, while $V_{i}(\alpha', i)$ denotes the
surface measure for regular spherical simplex with $(i+1)$ vertices
on the unit sphere $S^{i}$ and with inner product as $\alpha'$
between these $(i+1)$ vertices. (\ref{eq:internal}) is equal to
$B(\frac{1}{1+C^2k}, l-k)$, where
\begin{equation}
B(\alpha', m')=\theta^{\frac{m'-1}{2}} \sqrt{(m'-1)\alpha' +1}
\pi^{-m'/2} {\alpha'}^{-1/2}J(m',\theta) \label{eq:intB}
\end{equation}
with $\theta=(1-\alpha')/\alpha'$ and
\begin{equation}
 J(m', \theta)=\frac{1}{\sqrt{\pi}} \int_{-\infty}^{\infty}(\int_{0}^{\infty} e^{-\theta v^2+2i v\lambda} \,dv )^{m'} e^{-\lambda^2} \,d\lambda
\end{equation}
\label{lemma:internalapp}
\end{lemma}

\begin{proof}
Without loss of generality, assume that $F$ is a $(k-1)$-dimensional
face with $k$ vertices as $e_p, 1\leq p \leq k$, where $e_p$ is the
$n$-dimensional standard unit vector with the $p$-th element as `1';
and also assume that the $(l-1)$-dimensional face $G$ be the convex
hull of the $l$ vertices: $e_p, 1\leq p \leq k$ and $Ce_p, (k+1)\leq
p \leq l$. Then the cone $\text{Con}_{F,G}$ formed by observing the
$(l-1)$-dimensional face $G$ of the skewed cross-polytope
$\text{SP}$ from an interior point $x^F$ of the face $F$ is the
positive cone of the vectors:
\begin{equation}
Ce_j-e_i, ~~\text{for all}~~ j \in {J \backslash K},~i \in {K},
\label{eq:Grassstset}
\end{equation}
and also the vectors
\begin{equation}
e_{i_1}-e_{i_2}, ~~\text{for all}~~{i_1} \in {K},~{i_2} \in {K},
\label{eq:Grassndset}
\end{equation}
where $J=\{1, 2,..., l\}$ is the support set for the face $G$.

So the cone $\text{Con}_{F,G}$ is the direct sum of the linear hull
$L_{F}=\text{lin}\{F-x^F\}$ formed by the vectors in
(\ref{eq:Grassndset}) and the cone
$\text{Con}_{F^{\perp},G}=\text{Con}_{F,G} \bigcap L_{F}^{\perp}$,
where $L_{F}^{\perp}$ is the orthogonal complement to the linear
subspace $L_{F}$. Then $\text{Con}_{F^{\perp},G}$ has the same
spherical volume as $\text{Con}_{F,G}$.

Now let us analyze the structure of $\text{Con}_{F^{\perp},G}$. We
notice that the vector
\begin{equation*}
e_{0}=\sum_{p=1}^{k}e_p
\end{equation*}
is in the linear space $L_{F}^{\perp}$ and is also the only such a
vector (up to linear scaling) supported on $K$. Thus a vector $\x$
in the positive cone $\text{Con}_{F^{\perp},G}$ must take the form
\begin{equation}
-\sum_{i=1}^{k}{b_i \times e_i}+\sum_{i=k+1}^{l}{b_i \times e_i},
\label{eq:Grassvform}
\end{equation}
where $b_i, 1 \leq i \leq l$ are nonnegative real numbers  and
\begin{eqnarray*}
&&C \sum_{i=1}^{k}{b_i}=\sum_{i=k+1}^{l}{b_{i}}, \\
&&b_1=b_2=\cdots=b_k.
\end{eqnarray*}

That is to say, the  $(l-k)$-dimensional $\text{Con}_{F^{\perp},G}$
is the positive cone of $(l-k)$ vectors $a^1, a^2,...,a^{l-k}$,
where
\begin{equation*}
a^i=C\times e_{k+i}- \sum_{p=1}^{k}e_p/k,~~1\leq i \leq (l-k).
\end{equation*}

The normalized inner products between any two of these $(l-k)$
vectors is
\begin{equation*}
\frac{<a^i, a^j>}{\|a^i\|\|a^j\|}=\frac{k \times
\frac{1}{k^2}}{C^2+k \times \frac{1}{k^2}}=\frac{1}{1+kC^2}.
\end{equation*}
( In fact, $a^{i}$'s are also the vectors obtained by observing the
vertices $e_{k+1}, \cdots, e_{l}$ from $Ec=\sum_{p=1}^{k}e_p/k$, the
epicenter of the face $F$.)

We have so far reduced the computation of the internal angle to
evaluating (\ref{eq:interfrac}), the relative spherical volume of
the cone $\text{Con}_{F^{\perp},G}$ with respect to the sphere
surface $S^{l-k-1}$. This was computed as given in this lemma
\cite{Vershik92,henk} for the positive cones of vectors with equal
inner products by using a transformation of variables and the
well-known formula
\begin{equation*}
V_{i-1}(S^{i-1})=\frac{i\pi^{\frac{i}{2}}}{\Gamma(\frac{i}{2}+1)},
\end{equation*}
where $\Gamma(\cdot)$ is the usual Gamma function.

Instead, in this paper, we will give a proof of (\ref{eq:intB})
which can directly lead to the probabilistic large deviation method
of evaluating the internal angle exponent in \cite{Donoho06}.

First, we notice that $\text{Con}_{F^{\perp},G}$ is a
$(l-k)$-dimensional cone. Also, all the vectors $(x_1, \cdots,
x_{n})$ in the cone $\text{Con}_{F^{\perp},G}$ take the form in
(\ref{eq:Grassvform}).
 From \cite{Hadwiger79},
\begin{eqnarray}
&&\nonumber\int_{\text{Con}_{F^{\perp},G}}{e^{-\|x'\|^2}}\,dx'=\beta(F,G)
V_{l-k-1}(S^{l-k-1}) \\
&&\times \int_{0}^{\infty}{e^{-r^2}} r^{l-k-1}\,dr =\beta(F,G) \cdot
\pi^{(l-k)/2},
 \label{eq:Grassinaxchdirect}
\end{eqnarray}
where $V_{l-k-1}(S^{l-k-1})$ is the spherical volume of the
$(l-k-1)$-dimensional sphere $S^{l-k-1}$. Now define $U\subseteq
\mathbb{R}^{l-k+1}$ as the set of all nonnegative vectors satisfying:
\begin{equation*}
x_p \geq 0,~1 \leq p\leq l-k+1,~\sum_{p=2}^{l-k+1}x_p=Ckx_{1}
\end{equation*}
and define $f(x_1,~\cdots,~ x_{l-k+1}):U \rightarrow
\text{Con}_{F^{\perp},G}$ to be the linear and bijective map
\begin{eqnarray*}
f(x_1,~\cdots,~ x_{l-k+1})&=&-\sum_{p=1}^{k} x_1
e_p+\sum_{p=k+1}^{l} x_{p-k} \times e_p.
\end{eqnarray*}
Then
\begin{eqnarray}
&&~~\int_{ \text{Con}_{F^{\perp},G} }{e^{-\|x'\|^2}}\,dx' \nonumber \\
&&= \sqrt{\frac{l+(C^2-1)k}{C^2k}} \int_{\sum_{p=2}^{l-k+1}x_p =
Ckx_1, x_p \geq 0,~2 \leq p \leq l-k+1} \nonumber
\\&&~~~~~~~~~~~~~~~~~{e^{-\|f(x)\|^2}}\,d{x_2} \cdots
dx_{l-k+1} \nonumber \\
&&=\sqrt{\frac{l+(C^2-1)k}{C^2k}}\int_{\sum_{p=2}^{l-k+1}x_p =
Ckx_1, x_p \geq 0,~2 \leq p \leq l-k+1} \nonumber\\
&&~~~ e^{-kx_1^2-x_2^2-\cdots-x_{l-k+1}^{2} } \,d{x_2} \cdots
dx_{l-k+1} \label{eq:Grassvsigular}
\end{eqnarray}
where $\sqrt{\frac{l+(C^2-1)k}{C^2k}}$ is due to the change of
integral variables.

In fact, when
\begin{equation*}
x_p \geq 0,~1 \leq p\leq l-k+1,~\sum_{p=2}^{l-k+1}x_p=Ckx_{1},
\end{equation*}
the function $f$ is a linear transformation over the variables $x_2, ..., x_{l-k+1}$ with the following transformation matrix $M$ (disregarding the indices beyond $l$)
\[ \left( \begin{array}{ccccccc}
-\frac{1}{Ck} & \cdots & -\frac{1}{Ck}&1&0&\cdots&0\\
\vdots & \cdots & \vdots&0&1&\cdots&0\\
\vdots & \cdots & \vdots&\vdots&\vdots&\ddots&\vdots\\
-\frac{1}{Ck} & \cdots & -\frac{1}{Ck}&0&0&\cdots&1
 \end{array} \right).\]

It can then be calculated that the Jacobian of this transformation is $\sqrt{\det{MM^{T}}}=\sqrt{\frac{l+(C^2-1)k}{C^2k}}$, which accounts for the coefficient appearing in (\ref{eq:Grassvsigular}).

Now we define a random variable
\begin{equation*}
Z=X_2+X_3+\cdots+X_{l-k+1}-CkX_1,
\end{equation*}
where $X_1, X_2, \cdots, X_{n}$ are independent random variables,
with $X_p \sim HN(0,\frac{1}{2})$, $2 \leq p \leq (n-k+1)$, as
half-normal distributed random variables and $X_1\sim N(0,
\frac{1}{2k})$ as a normal distributed random variable. Then by
inspection, (\ref{eq:Grassvsigular}) is equal to
\begin{equation*}
 \frac{{\sqrt{\pi}}^{l-k+1}}{2^{l-k}} \times
 {\sqrt{l+(C^2-1)k}} p_{Z}(0).
\end{equation*}
where $p_{Z}(\cdot)$ is the probability density function for the
random variable $Z=X_2+X_3+\cdots+X_{l-k+1}-CKX_1$ and $p_{Z}(0)$ is
the probability density function $p_{Z}(z)$ evaluated at the
point $z=0$.

Use the notation
\begin{equation*}
G_{X}(\lambda)=\int_{-\infty}^{\infty} e^{i\lambda x}p_{X}(x) \,dx
\end{equation*}
as the characteristic function for any random variable $X$, where
$p_{X}(x)$ is the probability density function of $X$. Then from the
independence of $X_1, X_2,\ldots , X_{l-k+1}$, the characteristic
function for $Z$ is equal to
\begin{equation*}
G_{Z}(\lambda)=G_{X_{2}}(\lambda)^{l-k} \times G_{X_{1}}(\lambda).
\end{equation*}

Expressing the probability density function $p_{Z}(x)$ in the
Fourier domain, we have
\begin{equation*}
p_{Z}(0)=\frac{1}{2\pi}\int_{-\infty}^{\infty} G_{Z}(\lambda) \,
d\lambda.
\end{equation*}

Combining this with (\ref{eq:Grassinaxchdirect}), we can obtain the desired result (\ref{eq:intB}); and, as we already see in the derivations, we naturally arrive at the probabilistic explanation for the internal angle.
\end{proof}

So far at this point, we have only considered the internal angle $\beta(F,G)$ when $G$
is not the whole skewed cross-polytope. The following lemma
discusses the special case when $G=\text{SP}$.
\begin{lemma}
Suppose $F$, $K$ and $\text{SP}$ are defined in the same way as in
the statement of Lemma \ref{lemma:internalapp}. Then the internal angle
$\beta(F,\text{SP})$ between the $(k-1)$-dimensional face $F$ and
the $n$-dimensional skewed cross-polytope $\text{SP}$ is given by

\begin{eqnarray*}
 && \int_{-\infty}^{0}  \int_{-\infty}^{\infty} e^{-\frac{C^2k\lambda^2}{4}}  \left( \int_{0}^{\infty} e^{-\mu^2+i\mu\lambda}  \,d\mu \right)^{n-k} \\
 &&~~~~~~~~~~~~~~~~~~\frac{2^{n-k-1}}{{\sqrt{\pi}}^{n-k+2}}e^{-i\lambda z}  \,d\lambda     \,dz.
\end{eqnarray*}
\label{lemma:internalSP}
\end{lemma}

\begin{proof}
We use the same set of notations as in the proof of Lemma
\ref{lemma:internalapp}. Without loss of generality, assume $K=\{1,
\cdots,k\}$, $F$ is the $(k-1)$-dimensional face supported on $K$
and $G=\text{SP}$. So the cone $\text{Con}_{F,G}$ is the direct sum
of the linear hull $L_{F}=\text{lin}\{F-x^F\}$ formed by the vectors
in (\ref{eq:Grassndset}) and the cone
$\text{Con}_{F^{\perp},G}=\text{Con}_{F,G} \bigcap L_{F}^{\perp}$,
where $L_{F}^{\perp}$ is the orthogonal complement to the linear
subspace $L_{F}$. Then $\text{Con}_{F^{\perp},G}$ has the same
spherical volume as $\text{Con}_{F,G}$.

Following similar analysis in Lemma \ref{lemma:internalapp}, the cone
$\text{Con}_{F^{\perp},G}$ is the positive cone of $2(n-k)$ vectors
$a_{\pm}^1, a_{\pm}^2,...,a_{\pm}^{n-k}$, where
\begin{equation*}
a_{\pm}^i=\pm C\times e_{k+i}- \sum_{p=1}^{k}e_p/k,~~1\leq i \leq
(n-k).
\end{equation*}

This also means that $\text{Con}_{F^{\perp},G}$ is a
$(n-k+1)$-dimensional cone. Also, all the vectors $(b_1, \cdots,
b_{n})$ in the cone $\text{Con}_{F^{\perp},G}$ take the form
\begin{eqnarray*}
&&b_1=b_2=\cdots=b_{k}\leq 0,\\
&&\sum_{p=k+1}^{n}|b_{p}|\leq C k|b_1|.
\end{eqnarray*}

From \cite{Hadwiger79},
\begin{eqnarray}
&&\nonumber\int_{\text{Con}_{F^{\perp},G}}{e^{-\|x'\|^2}}\,dx'=\beta(F,G)
V_{n-k}(S^{n-k}) \\
&&\times \int_{0}^{\infty}{e^{-r^2}} r^{n-k}\,dr =\beta(F,G) \cdot
\pi^{(n-k+1)/2},
 \label{eq:inaxch}
\end{eqnarray}
where $V_{n-k}(S^{n-k})$ is the spherical volume of the
$(n-k)$-dimensional sphere $S^{n-k}$. Now define $U\subseteq
\mathbb{R}^{n-k+1}$ as the set of all the vectors taking the form:
\begin{equation*}
\{x_{1} \geq 0, \sum_{p=2}^{n-k+1}|x_p| \leq Ckx_{1}\}
\end{equation*}
and define $f(x_1,x_2,~\cdots,~ x_{n-k+1}):U \rightarrow
\text{Con}_{F^{\perp},G}$ to be the linear and bijective map
\begin{eqnarray*}
f(x_1,~\cdots,~ x_{n-k+1})&=&-\sum_{p=1}^{k} x_1
e_p+\sum_{p=k+1}^{n} x_{p-k+1} \times e_p.
\end{eqnarray*}
Then
\begin{eqnarray}
&&~~\int_{ \text{Con}_{F^{\perp},G} }{e^{-\|x'\|^2}}\,dx'= \sqrt{k} \int_{U}{e^{-\|f(x)\|^2}}\,dx \nonumber \\
&&=\sqrt{k} \int_{0}^{\infty} \int_{\sum_{p=2}^{n-k+1}|x_p| \leq
Ckx_1} \nonumber\\
&&~~~ e^{-kx_1^2-x_2^2-\cdots-x_{n-k+1}^{2} } \,d{x_2} \cdots
dx_{n-k+1} \,\,dx_1 \label{eq:sigular}
\end{eqnarray}
where $\sqrt{k}$ is due to the change of integral variables.

By inspection, (\ref{eq:sigular}) is equal to
\begin{equation*}
 {\sqrt{\pi}}^{n-k+1} P(X_2+X_3+\cdots+X_{n-k+1}-CkX_1 \leq 0),
\end{equation*}
where $X_1, X_2, \cdots, X_{n-k+1}$ are independent random variables,
with  $X_p \sim HN(0,\frac{1}{2})$, $2 \leq p \leq (n-k+1)$, as
half-normal distributed random variables and $X_1\sim N(0,
\frac{1}{2k})$ as a normal distributed random variable.

Expressing the probability density function of
$Z=X_2+X_3+\cdots+X_{n-k+1}-CKX_1$ in the Fourier domain, we can
simplify (\ref{eq:sigular}) to
\begin{eqnarray*}
 && \int_{-\infty}^{0}  \int_{-\infty}^{\infty} e^{-\frac{C^2k\lambda^2}{4}}  \left( \int_{0}^{\infty} e^{-\mu^2+i\mu\lambda}  \,d\mu \right)^{n-k} \\
 &&~~~~~~~~~~~~~~~~~~\frac{2^{n-k-1}}{\sqrt{\pi}}e^{-i\lambda z}  \,d\lambda     \,dz
\end{eqnarray*}
Combining this with (\ref{eq:inaxch}) gives us the desired result.
\end{proof}

\subsection{Derivation of the External Angles}
\begin{lemma}
Suppose that $F$ is a $(k-1)$-dimensional face of the skewed
cross-polytope
\begin{equation*}
\text{SP}=\{\y\in \mathbb{R}^n~|~\|\y_K\|_1+ \|\frac{\y_{\overline{K}}}{C}\|_1
\leq 1\}
\end{equation*}
supported on a subset $K$ with $|K|=k$. Then the external angle
$\gamma(G, \text{SP})$ between a $(l-1)$-dimensional face $G$ ($F
\subseteq G$) and the skewed cross-polytope $\text{SP}$ is given by
\begin{equation}
\gamma(G, \text{SP})=\frac{2^{n-l}}{{\sqrt{\pi}}^{n-l+1}}
\int_{0}^{\infty}e^{-x^2}(\int_{0}^{\frac{x}{C\sqrt{k+\frac{l-k}{C^2}}}}
e^{-y^2} \,dy)^{n-l}\,dx.
\end{equation}
\label{lemma:external}
\end{lemma}

\begin{proof}
Without loss of generality, assume $K=\{n-k+1, \cdots,n\}$. Consider
the $(l-1)$-dimensional face
\begin{equation*}
G=\text{conv}\{C\times e^{n-l+1}, ... ,C\times e^{n-k}, e^{n-k+1},
..., e^{n}\}
\end{equation*}
of the skewed cross-polytope $\text{SP}$. The $2^{n-l}$ outward
normal vectors of the supporting hyperplanes of the facets
containing $G$ are given by
\begin{equation*}
\{\sum_{p=1}^{n-l} j_{p}e_p/{C}+\sum_{p=n-l+1}^{n-k} e_p/C+
\sum_{p=n-k+1}^{n} e_p, j_{p}\in\{-1,1\}\}.
\end{equation*}

Then the outward normal cone $c(G, \text{SP})$ at the face $G$ is
the positive hull of these normal vectors. Thus
\begin{eqnarray}
&&\nonumber\int_{c(G,\text{SP})}{e^{-\|x\|^2}}\,dx=\gamma(G,SP)
V_{n-l}(S^{n-l}) \\
&&\times \int_{0}^{\infty}{e^{-r^2}}
r^{n-l}\,dx 
=\gamma(G,\text{SP}).\pi^{(n-l+1)/2},
 \label{eq:Grassaxch}
\end{eqnarray}
where $V_{n-l}(S^{n-l})$ is the spherical volume of the
$(n-l)$-dimensional sphere $S^{n-l}$. 
Now define $U$ to be the set
\begin{equation*}
\{x \in \mathbb{R}^{n-l+1} \mid x_{n-l+1} \geq 0, |x_p| \leq
\frac{x_{n-l+1}}{C}, 1\leq p \leq (n-l)\}
\end{equation*}
and define $f(x_1,~\cdots,~ x_{n-l+1}):U \rightarrow c(G,\text{SP})$
to be the linear and bijective map
\begin{eqnarray*}
f(x_1,~\cdots,~ x_{n-l+1})&=&\sum_{p=1}^{n-l} x_p
e_p+\sum_{p=n-l+1}^{n-k} \frac{x_{n-l+1}}{C}e_p \\
&~&+ \sum_{p=n-k+1}^{n} x_{n-l+1} \times e_p.
\end{eqnarray*}
Then
\begin{eqnarray*}
&&~~\int_{c(G,\text{SP})}{e^{-\|x'\|^2}}\,dx'\\
&&= \sqrt{k+\frac{l-k}{C^2}} \int_{U}{e^{-\|f(x)\|^2}}\,dx\\
&&=\sqrt{k+\frac{l-k}{C^2}} \int_{0}^{\infty}
\int_{-\frac{x_{n-l+1}}{C}}^{\frac{x_{n-l+1}}{C}}
\cdots\int_{-\frac{x_{n-l+1}}{C}}^{\frac{x_{n-l+1}}{C}}\\
&& e^{-x_1^2-\cdots
-x_{n-l}^2-(k+\frac{l-k}{C^2})x_{n-l+1}^{2} } \,dx_{1} \cdots \,dx_{n-l+1}\\
&&=\sqrt{k+\frac{l-k}{C^2}}\int_{0}^{\infty}
e^{-(k+\frac{l-k}{C^2})x^2}\\
&&~~~~~~~~~~~~~~~~~~~~\times \left(\int_{-\frac{x}{C}}^{\frac{x}{C}} e^{-y^2}\,dy \right )^{n-l} \,dx\\
&&=2^{n-l}\int_{0}^{\infty}
e^{-x^2} \left (\int_{0}^{ \frac{x}{C \sqrt{k+\frac{l-k}{C^2}} } } e^{-y^2}\,dy \right)^{n-l} \,dx, \\
\label{eq:axe}
\end{eqnarray*}
where $\sqrt{k+\frac{l-k}{C^2}}$ is due to the change of integral
variables. Combining it with (\ref{eq:Grassaxch}) leads to the
desired result.
\end{proof}

\subsection{Proof of Lemma \ref{lemma:existence}}
Consider any fixed $\delta>0$. First, we consider the internal angle
exponent $\psi_{int}$, where we define
$\gamma'=\frac{\rho\delta}{\frac{C^2-1}{C^2}\rho\delta+\frac{\nu}{C^2}}$.
Then for this fixed $\delta$,
\begin{equation*}
\frac{1-\gamma'}{\gamma'} \geq
\frac{\frac{C^2-1}{C^2}\rho\delta+\frac{\delta}{C^2}}{\rho\delta}-1
\end{equation*}
uniformly over $\nu \in [\delta, 1]$.

Now if we take $\rho$ small enough,
$\frac{\frac{C^2-1}{C^2}\rho\delta+\frac{\delta}{C^2}}{\rho\delta}$
can be arbitrarily large. By the asymptotic expression
(\ref{eq:intsmallasp}), this leads to large enough internal decay
exponent $\psi_{int}$. At the same time, the external angle exponent
$\psi_{ext}$ is lower-bounded by zero and the combinatorial exponent
is upper-bounded by some finite number.  Then if $\rho$ is small
enough, we will get the net exponent $\psi_{net}$ to be negative
uniformly over the range $\nu \in [\delta, 1]$.

\subsection{Proof of Lemma \ref{lemma:netasyp}}
We will show that for fixed $C>1$, with
$\rho{(\delta)}=\frac{1}{C^2}\log(1/\delta)^{-(1+\eta)}$ and some $\delta_0>0$, \
\begin{equation*}
\psi_{net}(\nu;\rho(\delta),\delta)<-\delta, ~~~~~\delta<\delta_{0},
~~~~\nu \in [\delta,1).
\end{equation*}
To this end, we need to get the asymptotic of $\psi_{int}(\nu)$,
$\psi_{ext}(\nu)$ and $\psi_{com}(\nu)$ as $\delta \rightarrow 0$
and $\rho{(\delta)}=\log(1/\delta)^{-(1+\eta)}$.

With
\begin{equation*}
H(\nu)+H(\rho\delta/\nu)\nu=H(\rho\delta)+H(\frac{\nu-\rho\delta}{1-\rho\delta})(1-\rho\delta),
\end{equation*}
from its definition, $\psi_{net}(\nu; \rho, \delta)$ is equal to
\begin{equation*}
H(\nu)-\psi_{ext}(\nu)-\xi_{\gamma'}(y_{\gamma'})(\nu-\rho\delta)
+\rho\delta\log(2)+H(\rho\delta/\nu)\nu.
\end{equation*}

From the derivation (or the expression) of the external angle
$\gamma(G,\text{SP})$ in this paper, $\gamma(G,\text{SP})$ is a
decreasing function in $C$. So we can upper-bound
$\gamma(G,\text{SP})$ uniformly in $\nu \in [\delta, 1]$, for any
$C\geq 1$, by
\begin{equation*}
\frac{2^{n-l}}{{\sqrt{\pi}}^{n-l+1}}
\int_{0}^{\infty}e^{-x^2}(\int_{0}^{\frac{x}{\sqrt{l}}} e^{-y^2}
\,dy)^{n-l}\,dx,
\end{equation*}
namely the expression for the external angle when $C=1$.

Now define $\Omega(\nu)=H(\nu)-\psi_{ext}^{C=1}(\nu)$,  where
$\psi_{ext}^{C=1}(\nu)$ is the external exponent when $C=1$. Then
from the asymptotic formula (\ref{eq:extasyp}), we have
\begin{equation*}
H(\nu)-\psi_{ext}(\nu) \leq \Omega(\nu)\sim
\frac{1}{2}\log(\log(\frac{1}{\nu}))\nu,
\end{equation*}
as $\nu \rightarrow 0$.

So $\psi_{net}(\nu; \rho, \delta)$ is no bigger than
\begin{equation*}
\Omega(\nu)-\xi_{\gamma'}(y_{\gamma'})(\nu-\rho\delta)
+\rho\delta\log(2)+H(\rho\delta/\nu)\nu. \label{eq:rdineq}
\end{equation*}

From \cite{Donoho06}, for a certain $\delta_{1}$, if
$\delta<\delta_{1}$,
\begin{equation*}
H(\rho\delta/\nu) \leq H(\rho)\delta+2\rho(\nu-\delta),
\end{equation*}
so we have
\begin{equation*}
\psi_{net}(\nu) \leq
K(\nu;\rho,\delta)+[\rho\delta\log(2)+H(\rho)\delta]+2\rho(\nu-\delta),
\end{equation*}
where $K(\nu;\rho,\delta) \doteq
\Omega(\nu)-\xi_{\gamma'}(y_{\gamma'})(\nu-\rho\delta)$.

As we will show later in Lemma \ref{lemma:concavity},
$K(\nu;\rho,\delta)$ is a concave function in $\nu \in [\delta, 1]$ if
$\delta<\delta_{2}$. Also, we will show that for $\delta<\delta_3$,
\begin{eqnarray}
&&K'(\delta;\rho,\delta) \leq -\eta/4 \log(\log(\frac{1}{\delta}))\label{eq:KK}\\
&&K(\delta;\rho,\delta) \leq -\delta \times \eta/4
\log(\log(\frac{1}{\delta})) \label{eq:K}
\end{eqnarray}
where $K'(\nu;\rho,\delta) \doteq \frac{\partial
K(\nu;\rho,\delta)}{\partial \nu}$. So there exists a $\delta_4>0$
so that for any $\delta\in (0, \delta_4)$,
\begin{equation*}
K'(\delta;\rho,\delta)<-2\rho.
\end{equation*}

Also there exists a $\delta_5>0$ so that for any
$0<\delta<\delta_5$,
\begin{equation*}
K(\delta;\rho,\delta)+\rho\delta\log(2)+H(\rho)\delta<-\delta.
\end{equation*}

Then by the concavity of $K(\nu;\rho,\delta)$, if
$\delta<\min{(\delta_1, \delta_2,\delta_3,\delta_4,\delta_5)}$,
\begin{equation*}
\psi_{net}(\nu) \leq  -\delta
\end{equation*}
uniformly over the interval $\nu \in [\delta, 1]$.

Now we need to prove (\ref{eq:KK}) and (\ref{eq:K}). As computed in
\cite{Donoho06},
\begin{equation*}
\Omega'(\nu)\sim\frac{1}{2}\log(\log(\frac{1}{\nu})),~~~~~\nu\rightarrow 0,
\end{equation*}

and

\begin{equation*}
\Omega(\nu)\sim \frac{1}{2}\log(\log(\frac{1}{\nu}))\nu, ~~~~~\nu\rightarrow 0.
\end{equation*}

By (\ref{eq:intsmallasp}), we know that as $\delta \rightarrow 0 $,
and with $\nu=\delta$,
\begin{equation*}
\xi_{\gamma'}(y_{\gamma'}) \sim \frac{1}{2} \log(\frac{1}{C^2 \rho})
\sim \frac{1}{2}\log(\log(\frac{1}{\delta})) (1+\eta).
\end{equation*}

Hence for $\delta<\delta_{6}$,
\begin{equation}
\xi_{\gamma'}(y_{\gamma'})(\nu-\rho\delta) \geq
(1+\frac{\eta}{2})\Omega(\delta).
\end{equation}

Following this, there exists a $\delta_{7}>0$ so that for
$\delta<\delta_7$,
\begin{equation*}
K(\delta;\rho, \delta)\leq \frac{\eta}{4} \log\log(\frac{1}{\delta})\delta.
\end{equation*}

Also, from the asymptotic of $\Omega'(\nu)$ and the asymptotic of
the derivative of $\xi_{\gamma'}(y_{\gamma'})(\nu-\rho\delta)$ with respect to $\nu$ in the next Lemma \ref{lemma:concavity}, we can further have
\begin{equation*}
K'(\delta;\rho,\delta) \leq -\eta/4 \log(\log(\frac{1}{\delta})).
\end{equation*}

\begin{lemma}
 If $\rho$ is small enough, $K(\nu;\rho,\delta)$ is concave as a function of $\nu$.

\begin{proof}
We define
\begin{equation*}
\Upsilon(\nu;\rho,\delta)=\xi_{\gamma'}(y_{\gamma'})(\nu-\rho\delta).
\end{equation*}

Since $K(\nu;\rho,\delta)=\Omega(\nu)-\Upsilon(\nu;\rho,\delta)$ and
$\Omega(\nu)$ is a concave function in $\nu$ \cite{Donoho06}, we only need to show
that $\Upsilon(\nu;\rho,\delta)$ is a convex function in $\nu$.

Recall that
$\gamma'=\frac{\rho\delta}{\frac{C^2-1}{C^2}\rho\delta+\frac{\nu}{C^2}}$
and we first look at $\frac{\partial \gamma'}{\partial \nu}$:
\begin{equation*}
\frac{\partial \gamma'}{\partial
\nu}=-\frac{\rho\delta}{\frac{C^2-1}{C^2}\rho\delta+\frac{\nu}{C^2}}
\times \frac{1}{(C^2-1)\rho\delta+\nu}.
\end{equation*}
So
\begin{eqnarray}
\frac{\partial\Upsilon(\nu;\rho,\delta)}{\partial
\nu}&=&\xi_{\gamma'}(y_{\gamma'})+\frac{\partial
\xi_{\gamma'}(y_{\gamma'})}{\partial \gamma'} \cdot\frac{\partial
\gamma'}{\partial \nu} \cdot (\nu-\rho\delta) \nonumber\\ &=&
\xi_{\gamma'}(y_{\gamma'})-\frac{\partial
\xi_{\gamma'}(y_{\gamma'})}{\partial \gamma'} \cdot \gamma' \cdot
\frac{1-\frac{\rho\delta}{\nu}}{(C^2-1)\frac{\rho\delta}{\nu}+1}\nonumber
\end{eqnarray}
If we define
$\Xi=\frac{1-\frac{\rho\delta}{\nu}}{(C^2-1)\frac{\rho\delta}{\nu}+1}$
and $\Pi=\frac{1}{(C^2-1)\rho\delta+\nu}$, we can have
\begin{eqnarray*}
&&\frac{\partial^2\Upsilon(\nu;\rho,\delta)}{\partial \nu^2} \\
&=& \frac{\partial \xi_{\gamma'}(y_{\gamma'})}{\partial \gamma'}
\frac{\partial \gamma'}{\partial \nu}
-\frac{\partial^2\xi_{\gamma'}(y_{\gamma'})}{\partial \gamma'^2}
\cdot \frac{\partial \gamma'}{\partial \nu}\cdot \gamma' \cdot \Xi \\
&&-\frac{\partial \xi_{\gamma'}(y_{\gamma'})}{\partial \gamma'}
\cdot \frac{\partial \gamma'}{\partial \nu} \cdot \Xi
-\frac{\partial \xi_{\gamma'}(y_{\gamma'})}{\partial \gamma'} \cdot
\frac{\partial \Xi}{\partial \nu} \cdot \gamma'\\
&=&\frac{\partial \xi_{\gamma'}(y_{\gamma'})}{\partial \gamma'}
\cdot (-\gamma') \cdot (1-\Xi)\cdot\Pi\\
&&-\frac{\partial \xi_{\gamma'}(y_{\gamma'})}{\partial \gamma'}
\cdot \gamma' \cdot \frac{\partial \Xi}{\partial \nu}
+\frac{\partial^2\xi_{\gamma'}(y_{\gamma'})}{\partial \gamma'^2}
\cdot \gamma'^2 \cdot \Xi \cdot \Pi
\\&=&-\frac{\partial \xi_{\gamma'}(y_{\gamma'})}{\partial \gamma'}
\cdot \gamma'^2 \cdot \Pi -\frac{\partial
\xi_{\gamma'}(y_{\gamma'})}{\partial \gamma'} \cdot \gamma'^2 \cdot
\Pi\\
&&+\frac{\partial^2\xi_{\gamma'}(y_{\gamma'})}{\partial \gamma'^2}
\cdot \gamma'^2 \cdot \Xi \cdot \Pi
\end{eqnarray*}

It has been shown in \cite{Donoho06} that as $\gamma \rightarrow 0$,
\begin{eqnarray*}
\frac{\partial \xi_{\gamma}{(y_{\gamma})}}{\partial \gamma} \sim
-\frac{\gamma^{-1}}{2},\\
\frac{\partial^2 \xi_{\gamma}{(y_{\gamma})}}{\partial \gamma^2} \sim
\frac{\gamma^{-4}}{4}.\\
\end{eqnarray*}

So from the definition of $\gamma'$, there exists a small enough
$\rho_{0}$ such that for any $\rho<\rho_{0}$,
\begin{equation}
\frac{\partial^2\Upsilon(\nu;\rho,\delta)}{\partial \nu^2}>0, \nu
\in [\delta, 1],
\end{equation}
which then implies the concavity of $K(\nu;\rho, \delta)$.
\end{proof}
\label{lemma:concavity}
\end{lemma}

\subsection{Proof of Lemma \ref{lemma:Csyp}}
\begin{proof}
Suppose instead that $\rho_{N}(\delta, C)> \frac{1}{C+1}$. Then for
every vector $\w$ from the null space of the measurement matrix $A$,
any $\rho_{N}(\delta,C)$ fraction of the $n$ components in $\w$ take no more
than $\frac{1}{C+1}$ fraction of $\|\w\|_1$. But this can not be
true if we consider the $\rho_N(\delta, C)$ fraction of $\w$ with
the largest magnitudes.

Now we only need to prove the lower bound for $\rho_{N}(\delta,C)$; in fact, we argue that
\begin{equation*}
\rho_{N}(\delta,C) \geq \frac{\rho_{N}(\delta,C=1)}{C^2}.
\end{equation*}
We know from Lemma \ref{lemma:existence} that $\rho_{N}(\delta,C)>0$
for any $C \geq 1$. Denote $\psi_{net}(C)$, $\psi_{com}(\nu;\rho,\delta,C)$, $\psi_{int}(\nu;\rho,\delta,C)$ and
$\psi_{ext}(\nu;\rho,\delta,C)$ as the respective exponents for a
certain $C$. Because $\rho_{N}(\delta, C=1)>0$, for any $\rho=\rho_{N}(\delta,C=1)-\epsilon$, where $\epsilon>0$ is an arbitrarily small number, the net exponent $\psi_{net}(C=1)$ is negative uniformly over $\nu \in [\delta, 1]$.

By examining the formula (\ref{eq:cncsex}) for the external angle
$\gamma(G, \text{SP})$, where $G$ is a $(l-1)$-dimensional face of
the skewed cross-polytope $\text{SP}$, we have $\gamma(G,\text{SP})$
is a decreasing function in both $k$ and $C$ for a fixed $l$. So
$\gamma(G,\text{SP})$ is upper-bounded by
\begin{equation}
\frac{2^{n-l}}{{\sqrt{\pi}}^{n-l+1}}
\int_{0}^{\infty}e^{-x^2}(\int_{0}^{\frac{x}{\sqrt{l}}} e^{-y^2}
\,dy)^{n-l}\,dx,
\end{equation}
namely the expression for the external angle when $C=1$. Then for
any $C>1$ and any $k$, $\psi_{ext}(\nu;\rho, \delta,C)$ is
lower-bounded by $\psi_{ext}(\nu;\rho, \delta,C=1)$.

Now let us check $\psi_{int}(\nu;\rho, \delta,C)$ by using the
formula (\ref{eq:inteptfor}). With
\begin{equation*}
\gamma'=\frac{\rho\delta}{\frac{C^2-1}{C^2}\rho\delta+\frac{\nu}{C^2}},
\end{equation*}
we have
\begin{equation}
\frac{1-\gamma'}{\gamma'}=-\frac{1}{C^2}+\frac{\nu}{C^2\rho\delta}.
\label{eq:fungamma}
\end{equation}

Then for any fixed $\delta>0$, if we take
$\rho=\frac{\rho_{N}(\delta, C=1)-\epsilon}{C^2}$, where $\epsilon$ is an arbitrarily small positive number, then for any $\nu\geq \delta$,
$\frac{1-\gamma'}{\gamma'}$ is an increasing function in $C$. So,
following easily from its definition,  $\xi_{\gamma'}(y_{\gamma'})$ is an increasing
function in $C$. This further implies that
$\psi_{int}(\nu;\rho,\delta)$ is an increasing function in $C$ if we
take $\rho=\frac{\rho_{N}(\delta, C=1)-\epsilon}{C^2}$, for any $\nu \geq \delta$.

Also, for any fixed $\nu$ and $\delta$, it is not hard to show that
$\psi_{com}(\nu;\rho,\delta,C)$ is a decreasing function in $C$ if
$\rho=\frac{\rho_{N}(\delta, C=1)}{C^2}$. This is because in (\ref{eq:Grassangformula}),
\begin{equation*}
\binom{n}{k}\binom{n-k}{l-k}=\binom{n}{l}\binom{l}{k}.
\end{equation*}

Thus for any $C>1$, if $\rho=\frac{\rho_{N}(\delta,C=1)-\epsilon}{C^2}$, the net exponent $\psi_{net}(C)$ is also negative uniformly over $\nu \in [\delta, 1]$. Since the parameter $\epsilon$ can be arbitrarily small, our claim and Lemma \ref{lemma:Csyp} then follow.
\end{proof}

\subsection{Proof of Lemma \ref{lemma:wthlimit}}
\begin{proof}
First, we notice that for any $C>1$,
\begin{eqnarray*}
&&\zeta_{W}(\delta) \geq \rho_{N}(\delta, C) \delta,\\
&&\zeta_{Sec}(\delta) \geq \rho_{N}(\delta, C) \delta,\\
&&\zeta_{S}(\delta) = \rho_{N}(\delta, C) \delta,\\
\end{eqnarray*}
so by Lemma \ref{lemma:existence},
\begin{eqnarray*}
&&\zeta_{W}(\delta) >0,\\
&&\zeta_{Sec}(\delta) >0,\\
&&\zeta_{S}(\delta) >0.
\end{eqnarray*}
Now we will prove
\begin{equation*}
\lim_{\delta \rightarrow 1} \zeta_{W}(\delta)=\delta.
\end{equation*}
As discussed in previous sections, we know that the decay exponent
for the probability that the condition (\ref{eq:Grasswthmeq1}) is
violated is equal to
\begin{equation*}
H(\frac{\nu-\rho\delta}{1-\rho\delta})(1-\rho\delta)-\xi_{\gamma'}(y_{\gamma'})(\nu-\rho\delta)-\psi_{ext}({\nu}).
\end{equation*}
But from the derivations of the exponents, we know that
\begin{equation*}
0=\lim_{\delta \rightarrow 1}{\sup_{\nu \in [\delta,1]}{\psi_{ext}({\nu})}},
\end{equation*}
\begin{equation*}
0=\lim_{\delta \rightarrow 1}{\sup_{\nu \in [\delta,1]}}{H(\frac{\nu-\rho\delta}{1-\rho\delta})},
\end{equation*}
\begin{equation*}
\xi_{\gamma'}(y_{\gamma'})(\nu-\rho\delta) \geq
\xi_{\gamma'(\delta)}(y_{\gamma'(\delta)})(1-\rho)\delta>0,~\nu \geq
\delta,
\end{equation*}
where
\begin{equation*}
\gamma'(\delta)=\frac{\rho\delta}{\frac{C^2-1}{C^2}\rho\delta+\frac{\delta}{C^2}}=\frac{\rho}{\frac{C^2-1}{C^2}\rho+\frac{1}{C^2}}.
\end{equation*}

Noticing that $\xi_{\gamma'(\delta)}(y_{\gamma'(\delta)})(1-\rho)>0$
is only determined by $\rho$, for any $0<\rho<1$, there exists a big
enough $\delta<1$, such that
\begin{equation*}
H(\frac{\nu-\rho\delta}{1-\rho\delta})(1-\rho\delta)-\xi_{\gamma'}(y_{\gamma'})(\nu-\rho\delta)-\psi_{ext}({\nu})<0.
\end{equation*}
uniformly over $[\delta, 1]$. Then it follows that
\begin{equation*}
\lim_{\delta \rightarrow 1} \zeta_{W}(\delta)=1,
\end{equation*}
for any $C>1$.

\end{proof}

\section*{Acknowledgment}
This work was supported in part by the National Science Foundation
under grant no. CCF-0729203, by the David and Lucille Packard
Foundation, and by Caltech's Lee Center for Advanced Networking.

\bibliographystyle{alpha}
\newcommand{\etalchar}[1]{$^{#1}$}

\end{document}

%% file: deflatex.tex
\def\real{{\mathchoice%
{\hbox{\rm\setbox1=\hbox{I}\copy1\kern-.45\wd1 R}}
{\hbox{\rm\setbox1=\hbox{I}\copy1\kern-.45\wd1 R}}
{\hbox{\scriptsize\rm\setbox1=\hbox{I}\copy1\kern-.45\wd1 R}}
{\hbox{\scriptsize\rm\setbox1=\hbox{I}\copy1\kern-.45\wd1 R}}}}
\def\Zint{{\mathchoice{\setbox1=\hbox{\sf Z}\copy1\kern-.75\wd1\box1}
{\setbox1=\hbox{\sf Z}\copy1\kern-.75\wd1\box1}
{\setbox1=\hbox{\scriptsize\sf Z}\copy1\kern-.75\wd1\box1}
{\setbox1=\hbox{\scriptsize\sf Z}\copy1\kern-.75\wd1\box1}}}
\newcommand{\complex}{ \hbox{\rm C\kern-0.45em\rule[.07em]{.02em}{.58em}%
\kern 0.43em}}

\newcommand{\be}{\begin{equation}}
\newcommand{\ee}{\end{equation}}
\newcommand{\beqr}{\begin{eqnarray}}
\newcommand{\eeqr}{\end{eqnarray}}
\newcommand{\beqrx}{\begin{eqnarray*}}
\newcommand{\eeqrx}{\end{eqnarray*}}
\newcommand{\ba}{\left[ \begin{array}}
\newcommand{\ea}{\\ \end{array} \right]}
\newcommand{\bi}{\begin{itemize}}
\newcommand{\ei}{\end{itemize}}

\newtheorem{lemma}{Lemma}
\newtheorem{theorem}{Theorem}

\newtheorem{definition}{Definition}